\documentclass{article}
\usepackage[utf8]{inputenc}

\usepackage[margin=1in]{geometry}
\usepackage{amsmath}
\usepackage{enumerate}
\usepackage{amsthm}
\usepackage{amssymb}
\usepackage{graphicx}
\usepackage{tikz}
\usepackage{pgfplots}
\usepgfplotslibrary{fillbetween}
\usetikzlibrary{arrows}
\usepackage{caption}
\usepackage{subcaption}
\usepackage{bbm}
\usepackage{enumitem}
\usepackage{etoolbox}
\usepackage[colorlinks=true]{hyperref}
\usepackage{float}
\usepackage{thmtools, thm-restate, lp}
\usepackage[autostyle, english = american]{csquotes}
\MakeOuterQuote{"}

\newtoggle{notblinded}
\settoggle{notblinded}{true}

\newtoggle{draft}
\settoggle{draft}{false}

\usetikzlibrary{calc}
\usetikzlibrary{decorations.pathreplacing}
\usetikzlibrary{shapes, patterns, decorations, fit, intersections, arrows, automata, positioning}
\usepackage{thmtools} 
\usepackage{hyperref}
\hypersetup{
    colorlinks=true,
    linkcolor=violet,
    filecolor=magenta,      
    urlcolor=cyan,
    citecolor=blue,
    pdffitwindow=true,
}\usepackage[capitalize, nameinlink]{cleveref}
\usepackage[backend=bibtex, style=alphabetic, backref=true,maxbibnames=99, url=false]{biblatex} 
\usepackage[noline,noend,linesnumberedhidden,boxed,nokwfunc]{algorithm2e}
\DontPrintSemicolon
\usepackage{algorithmic}

\usepackage[T1]{fontenc}

\newtheorem*{claim*}{Claim}

\newtheorem*{lemma*}{Lemma}

\newtheorem{prop}{Proposition}
\newtheorem*{prop*}{Proposition}
\newtheorem{theorem}{Theorem}
\newtheorem*{theorem*}{Theorem}
\newtheorem{defn}{Definition}
\newtheorem*{defn*}{Definition}

\newtheorem*{convention*}{Convention}

\newtheorem{fact}{Fact}
\newtheorem{remark}{Remark}

\iftoggle{draft}{
\newcommand{\bedit}[1]{\Authoredit{magenta}{#1}}
\newcommand{\red}[1]{\Authoredit{red}{#1}}

}
{
\newcommand{\bedit}[1]{#1}
\newcommand{\red}[1]{{#1}}

}

\def\pe{P_{\mathsf{even}}}
\def\po{P_{\mathsf{odd}}}
\def\pv{\mathbf{p}}
\def\qv{\mathbf{q}}

\definecolor{bleudefrance}{rgb}{0.19, 0.55, 0.91}
\definecolor{blizzardblue}{rgb}{0.67, 0.9, 0.93}

\theoremstyle{plain}

\usepackage{multicol}
\begin{document}

\title{A $\frac{4}{3}$-Approximation Algorithm\\for Half-Integral Cycle Cut Instances of the TSP}
\iftoggle{notblinded}{
\author{Billy Jin\\Cornell University \and Nathan Klein\\University of Washington \and David P.\ Williamson\\Cornell University}
}
{\author{}}

\pagenumbering{gobble}

\maketitle

\begin{abstract}
A long-standing conjecture for the traveling salesman problem (TSP) states that the integrality gap of the standard linear programming relaxation of the TSP (sometimes called the Subtour LP or the Held-Karp bound) is at most 4/3 for symmetric instances of the TSP obeying the triangle inequality; that is, the cost of an optimal tour is at most 4/3 times the value of the value of the corresponding linear program.  There is a variety of evidence in support of the conjecture (see, for instance, Goemans \cite{Goemans95} and Benoit and Boyd \cite{BenoitB08}).  It has long been known that the integrality gap is at most 3/2 (Wolsey \cite{Wolsey80}, Shmoys and Williamson \cite{ShmoysW90}).  Despite significant efforts by the community, the conjecture remains open.

In this paper we consider the half-integral case, in which a feasible solution to the LP has solution values in $\{0, 1/2, 1\}$.  Such instances have been conjectured to be the most difficult instances for the overall four-thirds conjecture \cite{SchalekampWvZ14}.  Karlin, Klein, and Oveis Gharan \cite{KKO20}, in a breakthrough result, were able to show that in the half-integral case, the integrality gap is at most 1.49993; Gupta et al.\ \cite{GuptaLLMNS22} showed a slight improvement of this result to 1.4983.   Additionally, this result led to the first significant progress on the overall conjecture in decades; the same authors showed the integrality gap of the Subtour LP is at most $1.5-\epsilon$ for some $\epsilon > 10^{-36}$ \cite{KKO21b}.

With the improvements on the 3/2 bound remaining very incremental, even in the half-integral case, we turn the question around and look for a large class of half-integral instances for which we can prove that the 4/3 conjecture is correct, preferably one containing the known worst-case instances.

In Karlin et al.'s work on the half-integral case, they perform induction on a hierarchy of critical tight sets in the support graph of the LP solution, in which some of the sets correspond to {\em cycle cuts} and the others to {\em degree cuts}.  Here we show that if all the sets in the hierarchy correspond to cycle cuts, then we can find a distribution of tours whose expected cost is at most 4/3 times the value of the half-integral LP solution; sampling from the distribution gives us a randomized 4/3-approximation algorithm.  We note that known bad cases for the integrality gap have a gap of 4/3 and have a half-integral LP solution in which all the critical tight sets in the hierarchy are cycle cuts; thus our result is tight.

Our overall approach is novel.  Most recent work has focused on showing that some variation of the Christofides-Serdyukov algorithm \cite{Christofides76, Serdyukov78} that combines a randomly sampled spanning tree plus a $T$-join (or a matching) can be shown to give a bound better than 1.5. 
Here we show that for any point in a region of ``patterns'' of edges incident to each cycle cut, we can give a distribution of patterns connecting all the child cycle cuts such that the distribution of patterns for each child also falls in the region.
This region gives rise to a distribution on Eulerian tours in which each edge in the support of the LP is used at most four-thirds of its LP value of the time, which then gives the result.

\end{abstract}

\newpage

\pagenumbering{arabic}
\setcounter{page}{1}

\section{Introduction}

In the traveling salesman problem (TSP), we are given a set of $n$ cities and the costs $c_{ij}$ of traveling from city $i$ to city $j$  for all $i,j$, and the goal of the problem is to find the least expensive tour that visits each city exactly once and returns to its starting point.
An instance of the TSP is called {\em symmetric} if $c_{ij} = c_{ji}$ for all $i,j$; it is {\em asymmetric} otherwise.  Costs obey the {\em triangle inequality} (or are {\em metric}) if $c_{ij} \leq c_{ik} + c_{kj}$ for all $i,j,k$.    For ease of exposition, we consider the problem input as a complete graph (undirected for symmetric instances, and directed for asymmetric instances) $G=(V,E)$ for the set of cities $V$, with $c_e = c_{ij}$ for edge $e=(i,j)$.  All instances we consider will be symmetric and obey the triangle inequality.

In the mid-1970s, Christofides \cite{Christofides76} and Serdyukov \cite{Serdyukov78} each gave a
$\frac{3}{2}$-approximation algorithm for the symmetric TSP with
triangle inequality.  The algorithm computes a minimum-cost spanning tree, and then finds a minimum-cost $T$-join on the odd degree vertices of the tree to compute a connected Eulerian subgraph. Because the edge costs satisfy the triangle inequality, this Eulerian subgraph can be ``shortcut'' to a tour of no greater cost.  Until very recently, this was the best approximation factor known for the symmetric TSP with triangle inequality, although over the last decade substantial progress was made for many special cases and variants of the problem; for example, in {\em graph TSP}, the input to the problem is an unweighted connected graph, and the cost of traveling between any two nodes is the number of edges in the shortest path between the two nodes.  A sequence of papers led to a 1.4-approximation algorithm for this problem due to Seb\H{o} and Vygen \cite{SeboV14}.  

In a breakthrough result, Karlin, Klein, and Oveis Gharan \cite{KKO21} gave the first approximation algorithm with performance ratio for the general case better than 3/2, although the amount by which the bound was improved is quite small (approximately $10^{-36}$).  The algorithm follows the Christofides-Serdyukov template by selecting a random spanning tree, then using a $T$-join on the odd degree vertices of the tree to create a connected Eulerian subgraph.  The achievement of the paper is to show that \red{sampling a random spanning tree from the max-entropy distribution}
gives a distribution of odd degree nodes in the spanning tree such that the expected cost of the $T$-join is cheaper (if marginally so) than in the Christofides-Serdyukov analysis.

One special case of the (symmetric, metric) TSP is known as the {\em half-integral} case.   To understand the half-integral case, we need to give a well-known LP relaxation of the TSP, first used by Dantzig, Fulkerson, and Johnson \cite{DantzigFJ54}, and sometimes called the {\em Subtour LP} or the {\em Held-Karp bound} \cite{HeldK71}.  The LP relaxation is as follows: 
\lps & &
& \mbox{Min} & \sum_{e \in E} c_{e} x_{e} \\
& \mbox{subject to:} & & & x(\delta(v)) = 2, & \forall v \in V,\\
& & & & x(\delta(S)) \geq 2, & \forall S\subset V, S \neq \emptyset,\\
& & & & 0 \leq x_{e} \leq 1, & \forall e \in E,  \elps 
where $\delta(S)$ is the set of all edges with exactly one endpoint in $S$ and we use the shorthand that $x(F) = \sum_{e \in F} x_e$.  A half-integral solution to the Subtour LP is one such that $x_e \in \{0, 1/2, 1\}$ for all $e \in E$, \red{and a half-integer instance of the TSP is one whose LP solution is half-integral}.

\begin{figure}
\begin{center}
\includegraphics[height=2in]{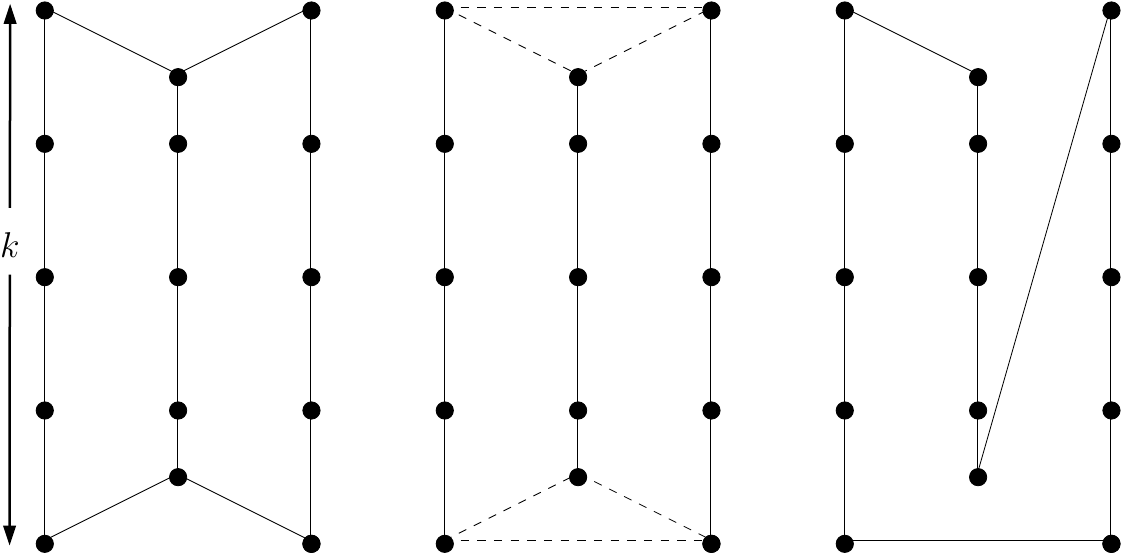}
\end{center}
\caption{Illustration of a known worst-case example for the integrality gap for
the symmetric TSP with triangle inequality.  The figure on the
left gives a graph, and costs $c_{ij}$ are the shortest path
distances in the graph.  The figure in the center gives the LP
solution, in which the dotted edges have value 1/2, and the solid
edges have value 1.  The figure on the right gives the optimal
tour.  The ratio of the cost of the optimal tour to the value of the LP solution tends to 4/3 as $k$ increases.} \label{fig:badexample}
\end{figure}

The {\em integrality gap} of an LP relaxation is the worst-case ratio of an optimal integer solution to the linear program to the optimal linear programming solution.  Wolsey \cite{Wolsey80} (and later Shmoys and Williamson \cite{ShmoysW90}) showed that the analysis of the Christofides-Serdyukov algorithm could be used to show that the integrality gap of the Subtour LP is at most 3/2.   It is known that the integrality gap of the Subtour LP is at least 4/3, due to a set of graph TSP instances shown in Figure \ref{fig:badexample}, and another set of weighted instances due to Boyd and Seb\H{o} \cite{BoydS21} known as $k$-donuts.  These instances are half-integral instances.  Schalekamp, Williamson, and van Zuylen \cite{SchalekampWvZ14} have conjectured that half-integral instances are the worst-case instances for the integrality gap.  It has long been conjectured\footnote{The first place that the authors are aware of a published statement of the conjecture is in a 1995 paper of Goemans \cite{Goemans95}, but the conjecture was in circulation earlier than that.} that the integrality gap is exactly 4/3, but until the work of Karlin et al.\ there had been no progress on that conjecture for several decades.  Goemans \cite{Goemans95} and Benoit and Boyd \cite{BenoitB08} give evidence that the 4/3 conjecture is correct.

In the case of half-integral instances, some results are known.  M\"omke and Svensson \cite{MomkeS16} have shown a 4/3-approximation algorithm for half-integral graph TSP, also yielding an integrality gap of 4/3 for such instances; because of the worst-case examples of Figure \ref{fig:badexample}, their result is tight.  Boyd and Carr \cite{BoydC11} give a 4/3-approximation algorithm (and an integrality gap of 4/3) for a subclass of half-integer solutions they call triangle points (in which the half-integer edges form disjoint triangles); the examples of Figure \ref{fig:badexample} show that their result is tight also.  Haddadan and Newman \cite{HaddadanN19} prove interesting results for the half-integral case with symmetric and metric costs.  Boyd and Seb\H{o} \cite{BoydS21} give an upper bound of 10/7 for a subclass of half-integral solutions they call square points (in which the half-integer edges form disjoint 4-cycles).  In a paper released just prior to their general improvement, Karlin, Klein, and Oveis Gharan \cite{KKO20} (KKO) gave a 1.49993-approximation algorithm in the half-integral case;  in particular, they show that given a half-integral solution, they can produce a tour of cost at most 1.49993 times the value of the corresponding objective function. Gupta, Lee, Li, Mucha, Newman, and Sarkar \cite{GuptaLLMNS22} improve this factor to 1.4983.  

With the improvements on the 3/2 bound remaining very incremental for weighted instances of the TSP, even in the half-integral case, we turn the question around and look for a large class of \red{weighted} half-integral instances for which we can prove that the 4/3 conjecture is correct, preferably one containing the known worst-case instances.

To define our instances, we turn to some terminology of KKO.  The KKO result uses induction on a hierarchy of {\em critical tight} sets of the half-integral LP solution $x$.  A set $S \subset V$ is {\em tight} if the corresponding LP constraint is met with equality; that is, $x(\delta(S))=2$.  A set $S$ is {\em critical} if it does not cross any other tight set; that is, for any other tight set $T$, either $S \cap T = \emptyset$ or $S \subseteq T$ or $T \subseteq S$.  The critical tight sets then give rise to a natural tree-like hierarchy based on subset inclusion.  

The KKO algorithm constructs a tour on the hierarchy by sampling a random spanning tree on the child nodes for each critical tight set, starting with the minimal sets in the hierarchy and working bottom up.  Following Christofides-Serdyukov, they then compute a minimum-cost $T$-join on the odd degree vertices of the resulting tree.  In their algorithm, they differentiate between {\em cycle cuts}  (in which the child nodes of a parent are linked by pairs of edges in a chain) and {\em degree cuts} (in which the child nodes of a parent form a 4-regular graph; more detail is given in subsequent sections). 

In this paper, we will consider half-integral instances in which there are only cycle cuts, which we will refer to as half-integral cycle cut instances.   Our contribution is to give a randomized $\frac{4}{3}$-approximation algorithm for these instances; it generates a distribution over connected Eulerian subgraphs with expected cost at most 4/3 the value of the LP solution.  More precisely, we give a distribution over connected Eulerian subgraphs such that each edge $e$ is used with expectation at most $\frac{4}{3}x_e$, which implies the result (note that edges are sometimes doubled in the Eulerian graph).  It is not hard to show that both the bad examples in Figure \ref{fig:badexample} and the $k$-donut instances of Boyd and Seb\H{o} \cite{BoydS21} are cycle cut instances (Boyd and Carr's result for triangle points works for the examples of Figure 1, but not for $k$-donuts). Thus our bound of 4/3 is tight and cannot be improved; furthermore, our result works for the known worst-case instances.

Our approach to the problem is novel and does not use the same Christofides-Serdyukov framework as employed by Karlin et al. \cite{KKO20} and others. 
Instead, we perform a top-down induction on the hierarchy of critical tight sets. For each set in the hierarchy, we define a set of ``patterns'' of edges incident on it such that the set has even degree. 
\red{
For each pattern, we give a distribution of edges connecting the chain of child nodes in the cycle cut, which induces a distribution of patterns on each child.  Crucially, we then show that there is a \emph{feasible region} $R$ of {distributions over patterns}, such that if the distribution of patterns on the parent node belongs to $R$, then the induced distribution on patterns on each child node also belong to $R$.
}


Our result leads to several interesting open questions.  The first is whether it is possible to extend the 4/3-approximation algorithm to the general half-integral case.  We believe it should be possible to improve modestly on the 1.4983-approximation of Gupta et al. \cite{GuptaLLMNS22} by combining our result for cycle cuts with some additional ideas.  We do not elaborate on this potential improvement because both the improvement and the additional ideas are incremental relative to the ideas introduced in this paper.  The second open question is whether our result extends to the case of cycle cuts for non-half-integral solutions. We believe this to be possible through a more refined understanding of the patterns that result from considering non-half-integral solutions.  A third open question is whether we can unify our result and that of Boyd and Carr on triangle points.  Triangle points need not be cycle cut instances, and it would be interesting to know of a single class of half-integral solutions that have an integrality gap of 4/3 and which captures both cycle cut instances and triangle points.

One major implication of our result is to focus attention on the half-integral degree cut case, in which every vertex has degree four, all edges have LP value 1/2, and every non-trivial cut has at least six edges in it.  While it is not clear whether a $\frac43$-approximation algorithm working on just these instances can be combined with our result for an overall $\frac43$-approximation algorithm for half-integral solutions, it is clear that understanding the degree cut case is a necessary next step to obtain any significant improvement in the approximation factor.  We believe that giving a feasible region on the distribution of patterns as described above will be useful in obtaining an improved approximation that integrates both degree and cycle cuts.

\subsection{Technical Overview}
We now give a more in-depth overview of our algorithm and proof techniques.  

Given a half-integral LP solution $x$, we construct a 4-regular 4-edge-connected multigraph $G$ by including every edge with $x_e=1/2$ once and every edge with $x_e=1$ twice. Therefore, in $G$ the tight sets $S$ have $|\delta(S)|=4$. For the remainder of this paper we will refer to this graph $G$ instead of a half integral solution $x$. 

Our strategy is to exhibit a convex combination of Eulerian tours that uses every edge at most $\frac23 = \frac43 x_e$ of the time. Using each edge of the graph at most $\frac23$ of the time immediately implies that when we sample a tour from this distribution, its expected cost will be at most $\frac{4}{3}$ times the cost of $x$.   This allows us to prove our main theorem.

\begin{theorem}
\label{thm:main}
There is a randomized $4/3$-approximation algorithm for half-integral cycle cut instances of the TSP that produces an Eulerian tour with expected cost at most $\frac43 \sum_{e \in E} c_ex_e.$
\end{theorem}

To construct our distribution of tours, we work on the cycle cut hierarchy from the top down. Each cut in the hierarchy is either a singleton vertex, or a cycle cut that contains two or more tight cuts inside.  
\red{At the root of the hierarchy is a cycle cut $S$ such that $V \setminus S$ is a single vertex.}
For every cycle cut, its child cuts are linked together in a chain, each cut connected to the next by a pair of edges.

\red{Our construction begins by specifying a distribution of \emph{patterns} entering the root cycle cut, where a "pattern" refers to a given multiset of edges that enter the cycle cut. }
We then work down the hierarchy to determine the distribution of patterns entering every cycle cut. Inductively, consider a cut $C$ in the hierarchy, and suppose we have already determined the distribution of patterns that enter $C$. We describe \emph{rules} that dictate how to connect the child cuts in the chain inside of $C$, as a function of the pattern that enters $C$. This in turn determines the distribution of patterns entering each child of $C$.  

\red{
The crux of the argument is to show a \emph{feasible region} $R$ of distributions over patterns, such that: 1) If the distribution of patterns on the parent node belongs to $R$, then the induced distribution on patterns on each child node also belong to $R$, and 2) the distributions in the region use each edge $e$ at most $\frac{4}{3} x_e$ of the time. We are able to impose any distribution we want on the root, thus this is sufficient to give the result.
}

Conceptually key to our analysis is a visualization of the process through the lens of \emph{Markov chains}. Each pattern corresponds to a state in the Markov chain. Given a distribution over patterns entering some cut $C$, applying the transition matrix gives the distribution of patterns entering the children of $C$.
\red{Using this language, the region $R$ will satisfy the property that $PR \subseteq R$, where $P$ is the transition matrix of the Markov chain. In fact, there will turn out to be two Markov chains, depending on if the parity of the number of cuts inside $C$ is even or odd; letting $\pe$ and $\po$ be the transition matrices of these two chains, our proof will show that $\pe \, R \subseteq R$ and $\po \, R \subseteq R$.} 

\red{
From the construction, one can easily sample an Eulerian tour in a top-down manner. To do so, we first choose an arbitrary distribution of patterns $\pv \in R$ from the feasible region. Then, we sample from $\pv$ a pattern entering the topmost cycle cut.
Next, supposing that we are at a cycle cut $C$ and have determined the pattern entering it, we follow the (randomized) rules to sample a set of edges to connect the children of $C$. Applying this procedure all the way down the hierarchy gives an Eulerian tour. By design, this (random) tour satisfies the property that the distribution of edges entering each cut in the hierarchy belongs to $R$, and that each edge is used at most $\frac43 x_e$ of the time in expectation.
}
\tikzset{
	pics/Graph/.style n args={1}{
	code = {
	
	\ifthenelse{#1>0}{
		\node[state] (u1) {};
        \node[state] (a1) [right=1cm of u1] {};
        \node[state] (a2) [below=0.5cm of a1] {};
        \node[state] (a3) [below=0.5cm of a2] {};
        \path[-] (u1) edge[midway, thick,below] node {} (a1);
		\path[-] (u1) edge[thick, bend right] node {} (a2);
		\path[-] (u1) edge[bend right] node {} (a3);
    }{
        \node[state] (u1) {};
        \node[state] (a1) {};
        \node[state] (a2) {};
        \node[state] (a3) {};
    }
	
	\ifthenelse{#1>0}{
	\foreach \a\b in {b/1,c/4,d/7} {
		\node[state] (\a0) [below right=0.6cm and {\b cm} of a1] {};
		\node[state] (\a1) [right=1cm of \a0] {};
        		\node[state] (\a2) [below=0.5cm of \a0] {};
        		\node[state] (\a3) [below=0.5cm of \a1] {};
        }
        }{
        \def\offset{2}
        \foreach \a\b\d in {b/1/,c/4/,d/7/} {
        		\foreach \c in {0,1,2} {
        			\node[state] (\a\c) [below right=0.6cm and {\b cm + \offset cm} of a1] {};
        		}
		\node[state] (\a3) [below right=0.6cm and {\b cm + \offset cm} of a1] {};
        }
        }
        
        \ifthenelse{#1>0}{
        \node[state] (e1) [right=11cm of u1] {};
        \node[state] (e2) [below=0.5cm of e1] {};
        \node[state] (e3) [below=0.5cm of e2] {};}{
        \foreach \a in {1,2} {
        		\node[state] (e\a) [right=12.2cm of u1] {};
        }
        \node[state] (e3) [right=12.2cm of u1] {};
        }
        
        \node[state] (l1) [right=12.2cm of u1] {};
        
	\foreach \a in {b,c,d} {
		\foreach \b in {1,2,3} {
			\path[-] (\a0) edge node {} (\a\b);
		}
		\path[-] (\a1) edge node {} (\a2);
		\path[-] (\a1) edge node {} (\a3);
		\path[-] (\a2) edge node {} (\a3);
	}
	
	\foreach \a in {a,e} {
		\path[-] (\a1) edge node {} (\a2);
		\path[-] (\a1) edge[bend right] node {} (\a3);
		\path[-] (\a2) edge[black] node {} (\a3);
	}
	
	\path[-] (u1) edge[midway,thick,bend left=20] node {} (l1);
	\path[-] (b1) edge node {} (c0);
	\path[-] (c3) edge[midway,thick] node {} (d2);
	\path[-] (a2) edge[midway,thick] node {} (b0);
	\path[-] (d1) edge node {} (e2);
	
	\path[-] (a1) edge[midway,thick,bend left=10] node {} (d0);
	
	\path[-] (l1) edge node {} (e1);
	\path[-] (l1) edge[bend left] node {} (e2);
	\path[-] (l1) edge[bend left] node {} (e3);
	
	\ifthenelse{#1>0}{
	\path[-] (b3) edge node {} (c2);
	\path[-] (a3) edge[midway,thick] node {} (b2);
	\path[-] (d3) edge node {} (e3);	
	\path[-] (c1) edge[midway,thick,bend left=10] node {} (e1);
	}{
	\path[-] (b3) edge[bend right=10] node {} (c2);
	\path[-] (a3) edge[below,midway,thick,bend right=10] node {} (b2);
	\path[-] (d3) edge[bend right=10] node {} (e3);
	\path[-] (c1) edge[below,midway,thick,bend left=10] node {} (e1);
	}
	}
	}
}

\tikzset{
	pics/intgap/.style n args={0}{
	code = {
	
    \node[state] (a0) {};
    \node[state] (b0) [below right=2cm and 0.75cm of a0] {};
    \node[state] (c0) [below=4cm of a0] {};
    
    \path[-] (a0) edge[thick,dotted] node {} (b0);
	\path[-] (a0) edge[thick,dotted] node {} (c0);
	\path[-] (b0) edge[thick,dotted] node {} (c0);
	
	\foreach \a/\b in {1/0,2/1,3/2,4/3,5/4} {
		\node[state] (a\a) [right=1.5cm of a\b]{};
		\path[-] (a\b) edge[thick] node {} (a\a);
		\node[state] (b\a) [right=1.5cm of b\b]{};
		\path[-] (b\b) edge[thick] node {} (b\a);
		\node[state] (c\a) [right=1.5cm of c\b]{};
		\path[-] (c\b) edge[thick] node {} (c\a);
    }
    \node[state] (a6) [right=1.5cm of a5]{};
	\path[-] (a5) edge[thick] node {} (a6);
	\node[state] (c6) [right=1.5cm of c5]{};
	\path[-] (c5) edge[thick] node {} (c6);
	
	\path[-] (a6) edge[thick,dotted] node {} (b5);
	\path[-] (a6) edge[thick,dotted] node {} (c6);
	\path[-] (b5) edge[thick,dotted] node {} (c6);
	}
	}
}

\tikzset{
	pics/intgapgraph/.style n args={0}{
	code = {
	
    \node[state] (a0) {};
    \node[state] (b0) [below right=2cm and 0.75cm of a0] {};
    \node[state] (c0) [below=4cm of a0] {};
    
    \path[-] (a0) edge[thick] node {} (b0);
	\path[-] (a0) edge[thick] node {} (c0);
	\path[-] (b0) edge[thick] node {} (c0);
	
	\foreach \a/\b in {1/0,2/1} {
		\node[state] (a\a) [right=1.5cm of a\b]{};
		\node[state] (b\a) [right=1.5cm of b\b]{};
		\node[state] (c\a) [right=1.5cm of c\b]{};
    }
    
    	\foreach \a/\b in {3/2} {
		\node[state] (a\a) [above right=1cm and 1.5cm of a\b]{};
		\node[state] (b\a) [right=1.5cm of b\b]{};
		\node[state] (c\a) [right=1.5cm of c\b]{};
    }
    
        	\foreach \a/\b in {4/3} {
		\node[state] (a\a) [below right=1cm and 1.5cm of a\b]{};
		\node[state] (b\a) [right=1.5cm of b\b]{};
		\node[state] (c\a) [right=1.5cm of c\b]{};
    }
    
    	\foreach \a/\b in {5/4} {
		\node[state] (a\a) [right=1.5cm of a\b]{};
		\node[state] (b\a) [right=1.5cm of b\b]{};
		\node[state] (c\a) [right=1.5cm of c\b]{};
    }
    
    	\foreach \a/\b in {1/0,2/1,3/2,4/3,5/4} {
		\path[-] (a\b) edge[thick,bend right=10] node {} (a\a);
		\path[-] (a\b) edge[thick,bend left=10] node {} (a\a);
		\path[-] (b\b) edge[thick,bend right=10] node {} (b\a);
		\path[-] (b\b) edge[thick,bend left=10] node {} (b\a);
		\path[-] (c\b) edge[thick,bend right=10] node {} (c\a);
		\path[-] (c\b) edge[thick,bend left=10] node {} (c\a);
    }

    \node[state] (a6) [right=1.5cm of a5]{};
	\path[-] (a5) edge[thick,bend right=10] node {} (a6);
	\path[-] (a5) edge[thick,bend left=10] node {} (a6);
	\node[state] (c6) [right=1.5cm of c5]{};
	\path[-] (c5) edge[thick,bend right=10] node {} (c6);
	\path[-] (c5) edge[thick,bend left=10] node {} (c6);
	
	\path[-] (a6) edge[thick] node {} (b5);
	\path[-] (a6) edge[thick] node {} (c6);
	\path[-] (b5) edge[thick] node {} (c6);
	}
	}
}
\section{Preliminaries}
\label{sec:prelim}

Given a half-integral LP solution $x$, we construct a 4-regular 4-edge-connected multigraph $G=(V,E)$ by including a single copy of every edge $e$ for which $x_e=\frac{1}{2}$ and two copies of every edge $e$ for which $x_e=1$. See \cref{fig:ex1} and \cref{fig:ex2} for examples.

\subsection{The structure of minimum cuts}

We state the following for general $k$-edge-connected multigraphs. In our setting, $k=4$.

\begin{figure}[htb]
\begin{center}
\begin{tikzpicture}[inner sep=1.7pt,scale=.8,pre/.style={<-,shorten <=2pt,>=stealth,thick}, post/.style={->,shorten >=1pt,>=stealth,thick}]
\draw [rotate=20,line width=1.1](0,0) ellipse (2cm and 1cm);
\draw [rotate=-20,line width=1.2](-1.35,-0.48) ellipse (2cm and 1cm);
\draw  (-3.8, 0.5) node {$X$};
\draw   (2.3, 0.5) node {$Y$};

\tikzstyle{every node} = [draw, circle,fill=red];
\node at (1.1,0.75)   (){};
\node at (0.6,0.3)   (){};
\path (1.3,0.1) node  (){};
\path (-1.9,0.35) node  (){};
\path (-2.5,0.5) node  (){};
\path (-0.5,0.2) node  (){};
\path (-1.2,-.7) node  (){};
\path (0,-0.7) node  (){};
\path (-1,1.5) node  (){};
\end{tikzpicture}
\end{center}
\caption{An example of two crossing sets.}
\label{fig:crossingsets}
\end{figure}
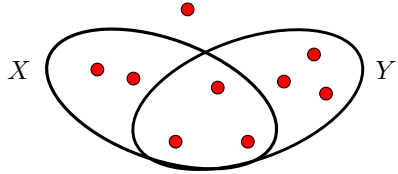

\begin{defn}\label{def:tightsets}
For a $k$-edge-connected multigraph $G=(V,E)$, we say:
\begin{itemize}
\item Any set $S\subseteq V$ such that $|\delta (S)|= k$ (i.e., its boundary is a minimum cut) is a \textbf{tight set}.
\item A set $S \subseteq V$ is \textbf{proper} if $2 \le |S| \le n-2$ and a \textbf{singleton} if $|S|=1$.
\item Two sets $S,S' \subseteq V$ \textbf{cross} if all of $S \smallsetminus S'$,
 $S' \smallsetminus S$, $S \cap S'$, and $V \smallsetminus (S \cup S') \not= \emptyset$ are non-empty (see \cref{fig:crossingsets}).
\end{itemize}
\end{defn}

\begin{figure}[!htbp]
\centering
\begin{tikzpicture}[
            auto,
            node distance = 2.5cm, 
            semithick 
        ]
        \tikzstyle{every state}=[
            draw = black,
            thick,
            fill = white,
            minimum size = 1.5mm
        ]
	\pic at (0,0) {Graph={1}};
	\draw[rotate=-40] [purple,line width=1.2pt, dashed] (1.13,0) ellipse (1.35 and 1.2);
	\foreach \a/\b in {B/3.3, C/6.3, D/9.3} {
	\draw [purple,line width=1.2pt, dashed] (\b+0.1,-1.3) ellipse (1.2 and 1);	}
	\draw [green,line width=1.2pt, dashed] (4.9,-0.8) ellipse (5.9 and 3.6);
	\draw [purple,line width=1.2pt, dashed] (5.3,-0.8) ellipse (6.5 and 3.8);
	\draw [green,line width=1.2pt, dashed] (3.6,-0.85) ellipse (4.4 and 3);
	\draw[rotate=-20] [blue,line width=1.2pt, dashed] (2.3,-0.2) ellipse (2.7 and 1.9);
	\draw [blue,line width=1.2pt, dashed] (5,-1.3) ellipse (2.85 and 1.6);
	\end{tikzpicture}
	    \caption{An example of a half integral instance, where the rightmost vertex is the root. The dotted lines circle (one side of) each minimum cut of the graph. The red cuts are "degree cuts" and the green cuts are the "cycle cuts." The blue cuts cross one another and therefore do not appear in the hierarchy of critical cuts.} \label{fig:ex1}
\end{figure}
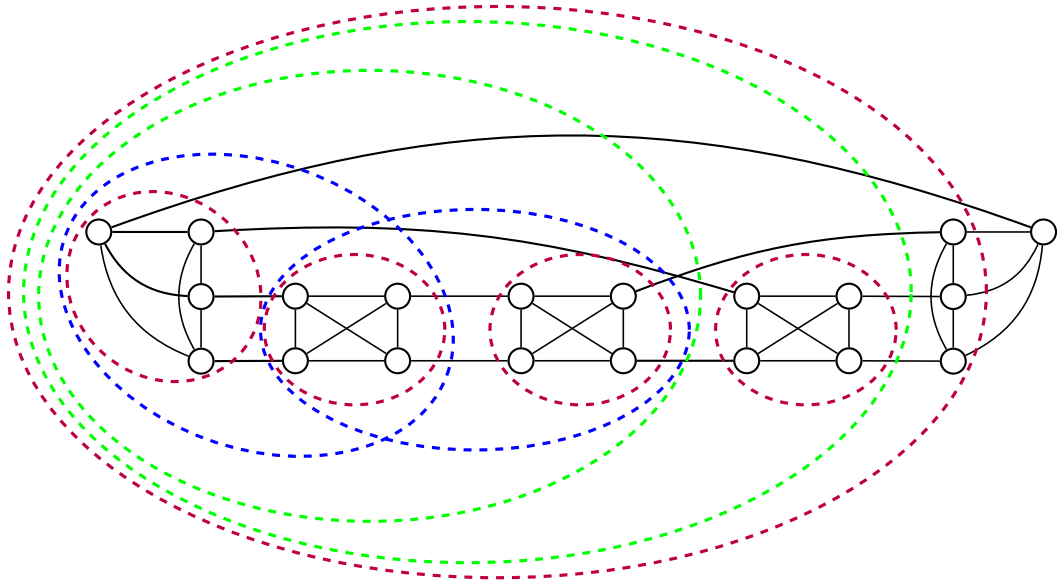



The following are two standard facts about minimum cuts; for proofs see~\cite{FF09}.
\begin{fact}
\label{fact:min-cuts1}
 If two tight sets $S$ and $S'$ cross, then each of $S \smallsetminus S'$,
 $S' \smallsetminus S$, $S \cap S'$ and $\overline{S \cup S'}$ are tight. Moreover, there are no edges from $S \smallsetminus S'$ to
 $S' \smallsetminus S$, and there are no edges from 
 $S \cap S'$ to $\overline{S \cup S'}$.
 \end{fact}

 \begin{fact}
 \label{fact:min-cuts2}
Let $G=(V,E)$ be a $k$-regular $k$-edge-connected graph. Suppose either $|V|=3$ or $G$ has at least one proper min cut, and every proper min cut is crossed by some other proper min cut. Then, $k$ is even and $G$ forms a cycle, with $k/2$ parallel edges between each adjacent pair of vertices.
\end{fact}

\subsection{Cycle cut instances and the hierarchy of critical tight sets}

We first define our class of instances.

\begin{defn}[Cycle cut instance]\label{def:cyclecut1}
	We say a graph $G$ is a \textbf{cycle cut instance} if every non-singleton tight set $S$ can be written as the union of two tight sets $A,B \not= S$. 
\end{defn}

Note that this definition includes complements of singletons, i.e. sets of size $n-1$. As mentioned in the introduction this condition captures the two known integrality gap examples of the subtour LP. See \cref{fig:ex2} for a cycle cut instance. Due to the following fact, an equivalent definition of a cycle cut instance is that all non-singleton tight sets which are \textit{not crossed} by any other tight set can be written as the union of two tight sets. 
\begin{fact}\label{fact:crossedtight}
	In any graph $G$, every tight set which is crossed by another tight set can be written as the union of two tight sets.
\end{fact}
\begin{proof}
	Let $S$ be a tight set crossed by $T$. Then by \cref{fact:min-cuts1}, $S \smallsetminus T$ and $S \cap T$ are tight sets; the claim follows.
\end{proof}
Thus, the condition in \cref{def:cyclecut1} is met trivially by all crossed sets. So, it is enough to ensure all sets that are not crossed can be written as the union of two tight sets.

We now show a third equivalent definition of cycle cut instances matching what is described in the introduction. First, fix an arbitrary \textbf{root vertex} $r \in V$, and for all cuts we consider we will take the side which does not contain $r$.




	
	

\begin{defn}[Critical cuts]
	A {\em critical} cut is any tight set $S \subseteq V \smallsetminus \{r\}$ which does not cross any other tight set. 
\end{defn}


\begin{defn}[Hierarchy of critical cuts, $\mathcal{H}$]
	Let $\mathcal{H} \subseteq 2^{V \smallsetminus r}$ be the set of all critical cuts.  
\end{defn}

For an example, see \cref{fig:ex1}. The hierarchy naturally gives rise to a parent-child relationship between sets as follows:
	
\begin{defn}[Child, parent, $E^\rightarrow(S)$]
	Let $S \in {\cal H}$ such that $|S| \ge 2$. Call the maximal sets $C \in {\cal H}$ for which $C \subset S$ the children of $S$, and call $S$ their parent.
	Finally, define $E^\rightarrow(S)$ to be the set of edges with endpoints in two different children of $S$. 
\end{defn}


	

\begin{defn}[Cycle cut, degree cut]\label{def:cyclecut}
	Let $S \in {\cal H}$ with $|S|\ge2$. Then we call $S$ a \textbf{cycle cut} if when $G \smallsetminus S$ and all of the children of $S$ are contracted, the resulting graph forms a cycle of length at least three with two parallel edges between each adjacent node. Otherwise, we call it a \textbf{degree cut}.
\end{defn} 

While this definition of a cycle cut may sound specialized, due to \cref{fact:min-cuts2}, cycle cuts arise very naturally from collections of crossing min cuts. See \cref{fig:ex1} for a general example whose hierarchy of critical tight sets contains both degree cuts and cycle cuts. 

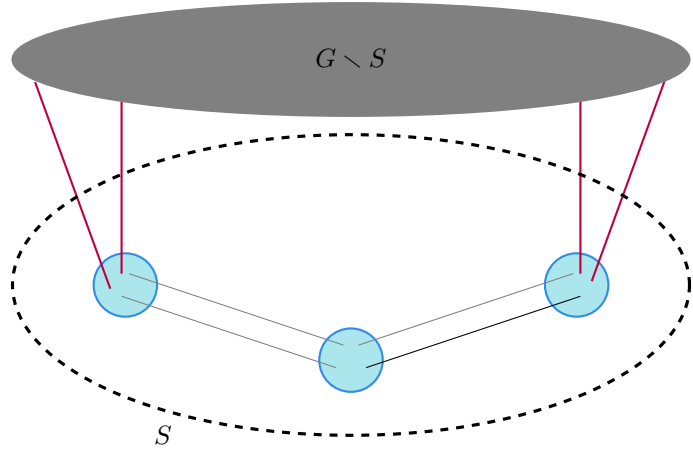
\begin{figure}[!htb]
\centering
\begin{tikzpicture}
\foreach \a/\b in {-3/1,0/0,3/1}{\draw[bleudefrance,fill=blizzardblue,thick] (\a,\b) circle (12pt);}

\draw[gray,xshift=-0.05cm,yshift=-0.15cm] (-3,1) to (-0.15,0.05); 
\draw[gray,xshift=0.05cm,yshift=--0.15cm] (-3,1) to (-0.15,0.05);

\draw[gray,xshift=-0.05cm,yshift=0.15cm] (0.15,0.05) to node [above] {} (3,1); 
\path[-] (0.15,0.05) edge[below,midway,,xshift=0.05cm,yshift=-0.15cm] node {} (3,1);

\path[-] (-3,1) edge[below,midway,thick,purple,xshift=-0.2cm,yshift=-0.05cm] node {} (-4,3.75);
\path[-] (-3,1) edge[right,midway,thick,purple,xshift=-0.05cm,yshift=0.15cm] node {} (-3,3.75);

\path[-] (3,1) edge[below,midway,thick,purple,xshift=0.2cm,yshift=0.05cm] node {} (4,3.75);
\path[-] (3,1) edge[right,midway,thick,purple,xshift=0.05cm,yshift=0.15cm] node {} (3,3.75);

\draw [black,line width=1.2pt, dashed] (0,1) ellipse (4.5 and 2);
\draw[gray,fill=gray,thick] (0,4) ellipse (4.5 and 0.75);

\node [draw=none] at (-2.5,-1) () {$S$};
\node [draw=none] at (0,4) () {$G \smallsetminus S$};

\end{tikzpicture}
\caption{$S$ is an example of a cycle cut with three children. In blue are contracted critical tight sets. In gray is the rest of the graph with $S$ contracted. As in \cref{fact:min-cuts2}, we can see that when $G \smallsetminus S$ is contracted into a single vertex, the resulting graph is a cycle with 2 edges between each adjacent vertex. In our recursive proof of our main theorem in \cref{sec:main}, we are given a distribution of Eulerian tours over $G/S$, so in particular on the red edges here, and will then extend it to $G$ with the blue critical sets contracted it by picking a distribution over the black edges.}
\label{fig:cyclecut}
\end{figure}


\begin{fact}
If $G$ is a cycle cut instance, then for any choice of $r$, ${\cal H}$ is composed only of cycle cuts (and singletons).
\end{fact}
\begin{proof}
    Let ${\cal H}$ be the hierarchy of critical cuts for an arbitrary choice of $r$ and let $S \in {\cal H}$ with $|S|\ge 2$ (note ${\cal H} \not= \emptyset$ since $V \smallsetminus \{r\} \in {\cal H}$). We will show that $S$ is a cycle cut. 
    
    Contract $G \smallsetminus S$ and all of the children of $S$; call the resulting graph $G'$. By definition of ${\cal H}$, $G'$ contains no proper tight sets which are not crossed. If $S$ contains a proper tight set which is crossed, then by \cref{fact:min-cuts2}, it is a cycle cut and we are done. 
    
    Otherwise, $G'$ contains no proper tight sets. Since  $G$ is a cycle cut instance and $S$ is a tight set, there exist tight sets $A,B$ such that $A \cup B = S$. Since the contracted children of $S$ are not crossed and there are no proper tight sets in $G'$, $A$ and $B$ must be vertices of this graph. Therefore, $G'$ must have exactly three vertices. Since it is 4-regular, it is a cycle of length three, and thus is a cycle cut.
    \end{proof}

Thus, in the remainder of the paper, we assume ${\cal H}$ is a collection of cycle cuts. We remark that the following is also true:
\begin{fact}
	If for some choice of $r$, ${\cal H}$ is composed only of cycle cuts, then $G$ is a cycle cut instance.
\end{fact}
\begin{proof}
Fix any tight set $S$ with $|S|\ge 2$. We will show it can be written as the union of two tight sets not equal to $S$. If it is crossed by another tight set, by \cref{fact:crossedtight} we are done. If $r \not\in S$, then $S$ appears in ${\cal H}$. Then, the graph in which the children of $S$ and $G \smallsetminus S$ are contracted is a cycle with vertices $a_0,\dots,a_k$. Then $S$ is the union of the vertices contained in tight sets $\{a_1,\dots,a_{k-1}\}$ and $\{a_k\}$.

Otherwise $r \in S$, and $\overline{S} \in {\cal H}$. Furthermore, since $|S|\ge 2$, $\overline{S} \not= S \smallsetminus \{r\}$. Thus, $\overline{S}$ has a parent $S'$. Consider the graph $a_0=G\smallsetminus S',a_1,\dots,a_k$ induced by contracting the children of $S'$ and $G \smallsetminus S'$. We have $\overline{S}=a_i$ for some $i \not= 0$. Then, $S$ can be written as $(V \smallsetminus \{a_i,a_{i-1}\}) \cup \{a_{i-1}\}$ as desired. 
\end{proof}

\begin{figure}
\centering
\begin{tikzpicture}
        \tikzstyle{every state}=[
            thick,fill = white,minimum size = 1.5mm
        ]
	\pic at (0,0) {intgapgraph};
	\draw [blue,line width=1.2pt, dashed] (5.6,-3.45) ellipse (6.9 and 2.3);
	\draw [blue,line width=1.2pt, dashed] (5.6,-2.25) ellipse (5.1 and 0.8);
	\draw [blue,line width=1.2pt, dashed] (5.6,-4.35) ellipse (6 and 0.8);
		\draw [blue,line width=1.2pt, dashed] (5.6,-4.35) ellipse (6 and 0.8);
	\draw [blue,line width=1.2pt, dashed] (5.6,-2.35) ellipse (8.5 and 3.4);
	\end{tikzpicture}
	    \caption{An example of a cycle cut instance, which is also the canonical integrality gap example for the Subtour LP. The non-singleton critical cycle cuts are shown in blue. The topmost vertex is the root.
	    } \label{fig:ex2}
\end{figure}
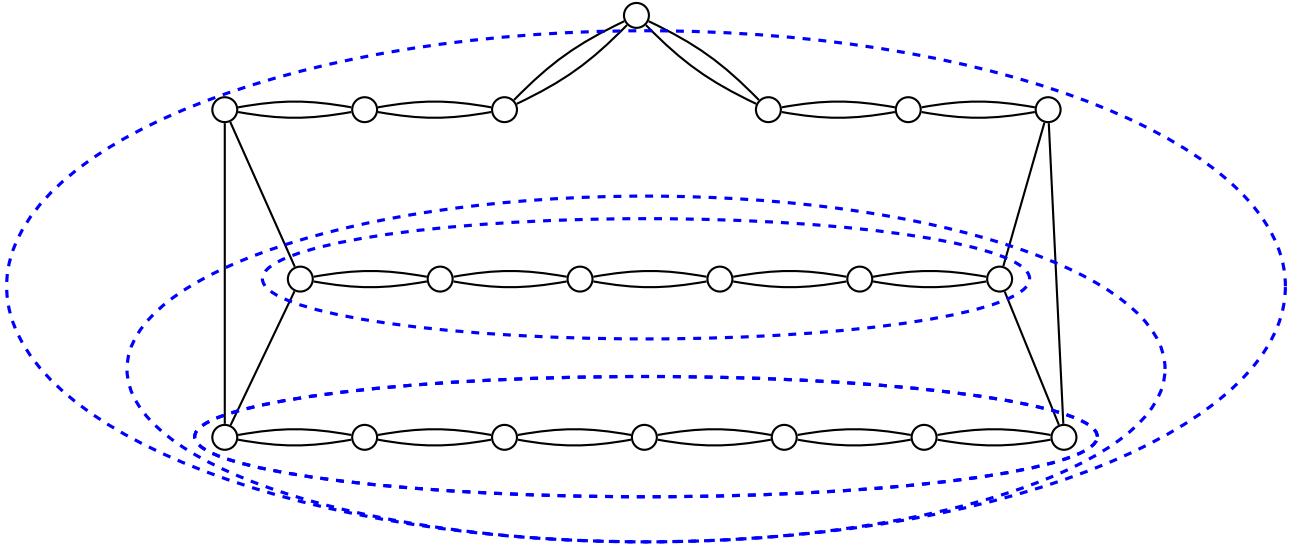

\subsection{Structure of cycle cuts}

 Given $S \in {\cal H}$, let $a_0 = G \smallsetminus S$ and let $a_1,\dots,a_k$ be its children in ${\cal H}$ (which are either vertices or cycle cuts). By \cref{fact:min-cuts2} $a_0,\dots,a_k$ can be arranged into a cycle such that two edges go between each adjacent vertex. WLOG let $a_1,\dots,a_k$ be in counterclockwise order starting from $a_0$. We call $a_1$ the leftmost child of $S$ and $a_k$ the rightmost child.
 

\begin{defn}[External and internal cycles cuts]
Let $S \in {\cal H}$ be a cut with parent $S'$. We call $S$ \textit{external} if in the ordering $a_0,\dots,a_k$ of $S'$ (as given above), $S=a_1$ or $S=a_k$. Otherwise, call $S$ \textit{internal}.
\end{defn}

For example, if the blue nodes in \cref{fig:cyclecut} are contracted cycle cuts, the left and right nodes are external, while the middle one is internal. Note that for an cycle cut $S$ with parent $S'$, if $S$ is external then $|\delta(S) \cap \delta(S')| = 2$, and if $S$ is internal then $|\delta(S) \cap \delta(S')| = 0$.

Using the following simple fact, we will now describe our convention for drawing and describing cycle cuts:
 
 \begin{fact}\label{fact:share<=1}
Let $A,B,C \in {\cal H}$ be three distinct critical cuts such that $A \subsetneq B$ and $B \cap C = \emptyset$ or $B \subseteq C$. Then $|\delta(A) \cap \delta(C)| \le 1$. 
 \end{fact}
\begin{proof}
Suppose otherwise, and $A$ shares two edges with $C$.

First, suppose $B \cap C = \emptyset$. Then, $A \cup C$ is a minimum cut, contradicting that $B$ was a critical cut since it is crossed by $A \cup C$. Note $A \cup C$ crosses $B$ since i) $B \smallsetminus (A \cup C) \not=\emptyset$ since $A \subsetneq B$, ii) $(A \cup C) \smallsetminus B \not=\emptyset$ since $C \cap B = \emptyset$, iii) $(A \cup C) \cap B \not= \emptyset$ since $A \subseteq B$ and finally iv) $(A \cup C) \cup B \not= V$ since neither contains the root. 

Otherwise, suppose $B \subseteq C$. But then $B \smallsetminus A$ has two edges to $A$ and two edges to $C \smallsetminus B$, implying that $C \smallsetminus A$ was a minimum cut. This contradicts that $B$ is a critical tight set, as it is crossed by $C \smallsetminus A$. We can verify $B$ crosses $C \smallsetminus A$ as follows: i) $B \smallsetminus (C \smallsetminus A) \not= \emptyset$ as it contains $A \not= \emptyset$, ii) $(C \smallsetminus A) \smallsetminus B \not= \emptyset$ as $B \subsetneq C$ and $A \subseteq B$, iii) $(C \smallsetminus A) \cap B \not= \emptyset$ since $B \smallsetminus A \not= \emptyset$, and finally iv) the union is not everything since neither contains the root. 
\end{proof}
 
\begin{defn}[$\delta^L(S), \delta^R(S)$]
	Let $S \in {\cal H}$ be a cycle cut. We will define a partition of $\delta(S)$ into two sets $\delta^L(S),\delta^R(S)$ each consisting of two edges.
	
	If $S \not= V \smallsetminus \{r\}$,
	then it has a parent $S'$. $S'$ has children $a_1,\dots,a_k$ such that $S=a_i$ for $i \not= 0$. 
	Let $\delta^L(S) = \delta(S) \cap \delta(a_{i-1})$ and $\delta^R(S) = \delta(a_{i+1 \pmod{k+1}}) \cap \delta(S)$. In other words, we partition the edges of $S$ into the two edges going to the left neighbor of $S$ in the cycle defined by $S'$'s children and the two edges going to the right neighbor.
	
	Otherwise $S = V \smallsetminus \{r\}$. Then if $a_1,\dots,a_k$ are the children of $S$, let $\delta^L(S)$ consist of an arbitrary edge from $\delta(S) \cap \delta(a_1)$ and an arbitrary edge from $\delta(S) \cap \delta(a_k)$. Let $\delta^R(S) = \delta(S) \smallsetminus \delta^L(S)$. 
	

\end{defn}

By \cref{fact:share<=1} and the definition of $\delta^L(S),\delta^R(S)$ for $S=V \smallsetminus \{r\}$, if $S'$ is an external child of a cycle cut $S$, then $|\delta^L(S') \cap \delta(S)| = |\delta^R(S') \cap \delta(S)| = 1$. 
This allows us to adopt the following convention for drawing cycle cuts which we will call the \textbf{caterpillar drawing} of $S$: for an example, see \cref{fig:caterpillar}. Formally, let $S \in {\cal H}$ be a cycle cut with children $a_1,\dots,a_k \in {\cal H}$. Arrange $a_1,\dots,a_k$ in a horizontal line. First, expand the children of $a_1$ vertically (if it is not a singleton) such that the unique edge in $\delta^L(S) \cap \delta(a_1)$ is pointing up (if it is a singleton, simply draw this edge pointing up. Then, expand $a_2,\dots,a_k$ one by one into their respective children (if they exist), placing the children vertically in increasing or decreasing order of their index so that the edges from $a_i$ to $a_{i+1}$ do not cross. If $a_k$ is a singleton, arbitrarily choose which edge to draw pointing up. Otherwise, let $a'$ be the topmost child of $a_k$. Draw the unique edge in $\delta(S) \cap \delta(a')$ pointing up.

There are two distinct types of cycle cuts:

\begin{defn}[Straight and twisted cycle cuts]\label{def:intersectingcyclecut}
    Let $S \in {\cal H}$ be a cycle cut. caterpillar drawing of $S$. If $\delta^L(S)$ has both edges pointing up, then call it a \textit{straight} cycle cut. Otherwise, call it a \textit{twisted} cycle cut. See \cref{fig:caterpillar} for examples. 
\end{defn}

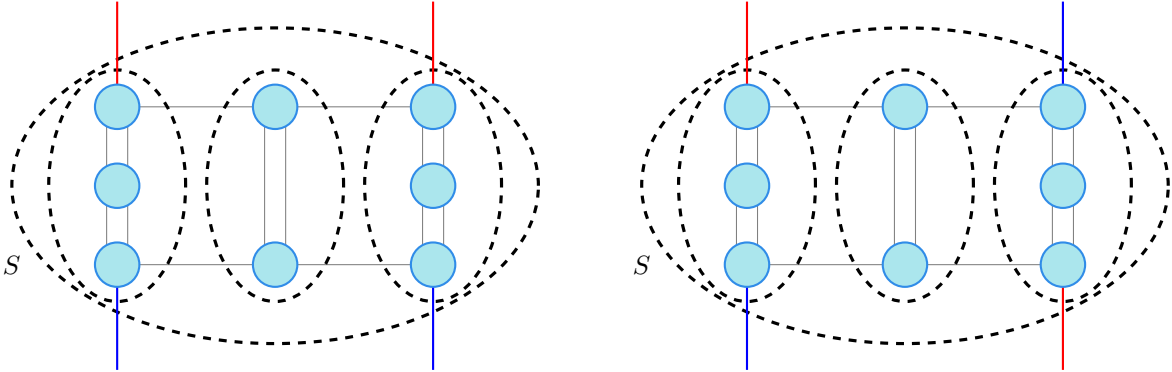
\begin{figure}[!htbp]
\centering
\begin{tikzpicture}[scale=0.7]

\foreach \a/\b in {-3/-1.5,-3/0.1, 3/-1.5,3/0.1}{
\draw[gray,xshift=-0.2cm] (\a,\b) to (\a,\b+1.4);
\draw[gray,xshift=0.2cm] (\a,\b) to (\a,\b+1.4);
}
\draw[gray,xshift=-0.2cm] (0,-1.5) to (0,1.5);
\draw[gray,xshift=0.2cm] (0,-1.5) to (0,1.5);

\draw[gray] (-3,1.5) to (-0.4,1.5);
\draw[gray] (-3,-1.5) to (-0.4,-1.5);

\draw[gray] (0,1.5) to (3,1.5);
\draw[gray] (0,-1.5) to (3,-1.5);

\draw [black,line width=1.2pt, dashed] (-3,0) ellipse (1.3 and 2.2);
\draw [black,line width=1.2pt, dashed] (0,0) ellipse (1.3 and 2.2);
\draw [black,line width=1.2pt, dashed] (3,0) ellipse (1.3 and 2.2);

\draw [black,line width=1.2pt, dashed] (0,0) ellipse (5 and 3);

\node [draw=none] at (-5,-1.5) () {$S$};

\path[-] (-3,1.5) edge[below,midway,thick,red] node {} (-3,3.5);
\path[-] (3,1.5) edge[below,midway,thick,red] node {} (3,3.5);

\path[-] (-3,-1.5) edge[below,left,thick,blue] node {} (-3,-3.5);
\path[-] (3,-1.5) edge[below,right,thick,blue] node {} (3,-3.5);

\foreach \a/\b in {-3/-1.5,-3/0,-3/1.5, 0/-1.5,0/1.5, 3/-1.5,3/0,3/1.5}{
\draw[bleudefrance,fill=blizzardblue,thick] (\a,\b) circle (12pt); 
}

\end{tikzpicture}\hspace{10mm}
\begin{tikzpicture}[scale=0.7]

\foreach \a/\b in {-3/-1.5,-3/0.1, 3/-1.5,3/0.1}{
\draw[gray,xshift=-0.2cm] (\a,\b) to (\a,\b+1.4);
\draw[gray,xshift=0.2cm] (\a,\b) to (\a,\b+1.4);
}
\draw[gray,xshift=-0.2cm] (0,-1.5) to (0,1.5);
\draw[gray,xshift=0.2cm] (0,-1.5) to (0,1.5);

\draw[gray] (-3,1.5) to (-0.4,1.5);
\draw[gray] (-3,-1.5) to (-0.4,-1.5);

\draw[gray] (0,1.5) to (3,1.5);
\draw[gray] (0,-1.5) to (3,-1.5);

\draw [black,line width=1.2pt, dashed] (-3,0) ellipse (1.3 and 2.2);
\draw [black,line width=1.2pt, dashed] (0,0) ellipse (1.3 and 2.2);
\draw [black,line width=1.2pt, dashed] (3,0) ellipse (1.3 and 2.2);

\draw [black,line width=1.2pt, dashed] (0,0) ellipse (5 and 3);

\node [draw=none] at (-5,-1.5) () {$S$};

\path[-] (-3,1.5) edge[below,midway,thick,red] node {} (-3,3.5);
\path[-] (3,1.5) edge[below,midway,thick,blue] node {} (3,3.5);

\path[-] (-3,-1.5) edge[below,left,thick,blue] node {} (-3,-3.5);
\path[-] (3,-1.5) edge[below,right,thick,red] node {} (3,-3.5);

\foreach \a/\b in {-3/-1.5,-3/0,-3/1.5, 0/-1.5,0/1.5, 3/-1.5,3/0,3/1.5}{
\draw[bleudefrance,fill=blizzardblue,thick] (\a,\b) circle (12pt); 
}
\end{tikzpicture}
\caption{Two caterpillar drawing of two different cycle cuts $S$ with three children. The red edges are in the $\delta^L(S)$ partition, and the blue edges are in the $\delta^R(S)$ partition. The left drawing is a straight cycle cut, and the right is a twisted cycle cut as per  \cref{def:intersectingcyclecut}.}
\label{fig:caterpillar}
\end{figure}

In future sections, we abbreviate the caterpillar drawing by contracting the non-singleton children of $S$. We do so partially for cleaner pictures but also to emphasize that all the relevant information used by our construction in the following section is contained in the abbreviated pictures.

\begin{figure}[!htbp]
\centering
\begin{tikzpicture}[scale=0.7]
\draw[gray] (-3,0.6) to (0,0.6);
\draw[gray] (-3,-0.6) to (0,-0.6);
\draw[gray] (0,0.6) to (3,0.6);
\draw[gray] (0,-0.6) to (3,-0.6);
\draw [black,line width=1.2pt, dashed] (0,0) ellipse (4.2 and 1.7);
\node [draw=none] at (-4,-1) () {$S$};
\path[-] (-3,0.5) edge[below,midway,thick,red] node {} (-3,2);
\path[-] (3,0.5) edge[below,midway,thick,red] node {} (3,2);
\path[-] (-3,-0.5) edge[below,left,thick,blue] node {} (-3,-2);
\path[-] (3,-0.5) edge[below,right,thick,blue] node {} (3,-2);
\foreach \a/\b in {-3/0,0/0,3/0}{
\draw[bleudefrance,fill=blizzardblue,thick] (\a,\b) circle (24pt); 
}
\end{tikzpicture}\hspace{10mm}
\begin{tikzpicture}[scale=0.7]
\draw[gray] (-2.4,0.6) to (-0.6,0.6);
\draw[gray] (-2.4,-0.6) to (-0.6,-0.6);
\draw[gray] (0.6,0.6) to (2.4,0.6);
\draw[gray] (0.6,-0.6) to (2.4,-0.6);
\draw [black,line width=1.2pt, dashed] (0,0) ellipse (4.2 and 1.7);
\node [draw=none] at (-4,-1) () {$S$};

\path[-] (-3,0.5) edge[below,midway,thick,red] node {} (-3,2);
\path[-] (3,0.5) edge[below,midway,thick,blue] node {} (3,2);
\path[-] (-3,-0.5) edge[below,left,thick,blue] node {} (-3,-2);
\path[-] (3,-0.5) edge[below,right,thick,red] node {} (3,-2);
\foreach \a/\b in {-3/0,0/0,3/0}{
\draw[bleudefrance,fill=blizzardblue,thick] (\a,\b) circle (24pt); 
}
\end{tikzpicture}
\caption{On the left is a shorthand caterpillar drawing for the straight cycle cut on the left in \cref{fig:caterpillar}, the style of picture we will use in future sections. To obtain this picture we contract the children of the cut. These are the shorthand drawings of \cref{fig:caterpillar}, so the left is a straight cycle cut and the right is twisted. We also label the edges as described below.}
\label{fig:short-caterpillar}
\end{figure}
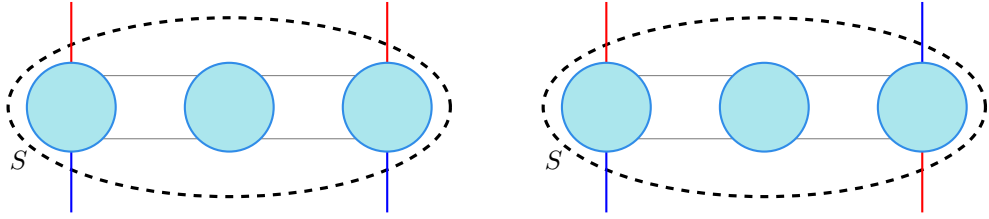

To help the reader's understanding, we suggest looking at \cref{fig:ex2}. The largest critical cut is arbitrarily chosen to be straight or twisted. The second largest critical cut is a straight cut. The smallest non-singleton critical cuts are arbitrarily chosen to be straight or twisted. 



In the following, we will need to distinguish between the straight and twisted types of cuts as well as those with an even versus an odd number of children.

\section{Proof of \Cref{thm:main}}\label{sec:main}

We now present the proof of our main result, a $\frac43$-approximation for half-integral  cycle-cut instances of the TSP.
To prove Theorem \ref{thm:main}, we construct a distribution of Eulerian tours such that every edge is used at most $\frac{2}{3}$ of the time. Since $x_e = \frac{1}{2}$ for every edge in the graph, this immediately implies that when we sample a tour from this distribution, its expected cost is at most $\frac{4}{3}$ times the value of the LP. We work on the cycle cut hierarchy from the top down, and inductively specify the distribution of edges that enter every cut. 

\Cref{fig:short-caterpillar} depicts our convention for visualizing a cycle cut as described in \cref{sec:prelim}. We say that a cycle cut is \emph{even} if it contains an even number of children, and \emph{odd} otherwise. 
\Cref{fig:patterns} illustrates the \emph{patterns} we use, where "pattern" refers to a multiset of edges that enter a cycle cut. For each pattern entering a parent cycle cut, we give (randomized) rules which describe how to connect up its children -- this induces a distribution of patterns for each child.
We represent this process using a Markov chain with 4 states, illustrated in \Cref{fig:patterns}. The figure shows the mapping from patterns to states; the transitions will come from the rules for connecting up the children, which we describe later. In the figure, each state contains two pictures, which represent the \emph{parity} of the edges in the patterns that are mapped to the state. Specifically, an edge that is present is used exactly once, whereas an edge that is not present may be either unused or doubled.
For example, \Cref{fig:doubled} illustrates all possible patterns that are captured by the top picture of state 1. Finally, we maintain the invariant that if a cycle cut is in a given state, then each of the two pictures are equally likely. (When we later give the rules for connecting up the children, we will ensure this invariant is preserved.) Thus, when we say a cycle cut is in a given state with probability $p$, this means the parity of the pattern entering it follows the top picture in the state with probability $\frac{p}{2}$, and the bottom picture with probability $\frac{p}{2}$. We will use the phrase "the distribution of patterns entering a cycle cut $C$ is $(p_1, p_2, p_3, p_4)$" to mean that for all $i \in \{1,2,3,4\}$, $C$ is in state $i$ with probability $p_i$. 

To prove our main result, we give a \emph{feasible region} $R$ of distributions over the states of the Markov chain,
and show that: 1)  If the distribution of patterns entering a cycle cut $C$ belongs to $R$, there is a way to connect up the children of $C$ such that the distribution on each child also belongs to $R$, and 2) for each $\pv \in R$, the corresponding rule for connecting the children of $C$ uses each edge in $E^{\rightarrow}(C)$ at most $\frac23 = \frac43 x_e$ of the time in expectation. The feasible region is given in \Cref{def:region}. Since $R$ is nonempty, 1) and 2) are sufficient to give the result since we can induce any distribution 
on the topmost cycle cut $V \smallsetminus \{r\}$.

\begin{defn}[The Feasible Region]
\label{def:region}
Let 
$$R = \left\{(p_1, p_2, p_3, p_4) \in \mathbb{R}^4_+: p_1+p_2+p_3+p_4 = 1, \, p_1 + p_2 = \frac23,\, p_2 + p_4 \geq \frac13\right\}.$$
See \Cref{fig:region} for a visualization of $R$ in a 2-dimensional space.
\end{defn}

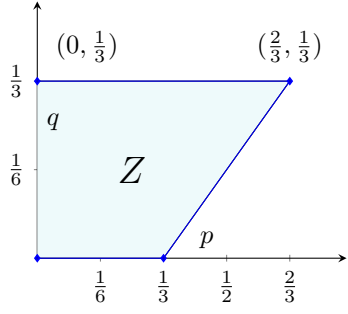
\begin{figure}[ht]
    \centering
    \begin{tikzpicture}[auto, semithick, black, node distance=3.5cm, scale=0.6]
    \begin{axis}[
    axis lines = middle,
    xlabel={$p$},
    ylabel={$q$},
    xmin = 0,
    xmax = 4.9,
    ymin = 0, 
    ymax = 2.9,
    xtick={0, 1, 2, 3, 4},
    xticklabels={0, $\frac16$, $\frac13$, $\frac12$, $\frac23$},
    ytick={0, 1, 2},
    yticklabels={0, $\frac16$, $\frac13$}
    ]
    \addplot[domain=2:4, name path=A] {x-2};
    \addplot[domain=0:4, name path=C] {2};
    \path[name path=B] (0,0) -- (0,2);
    \addplot[domain=0:2, name path=D] {0};
    
    \addplot[blizzardblue!20] fill between[of=C and D];
    
    \addplot coordinates {(0,0) (2,0) (4,2) (0,2)};
    
    \node[label=above right:{$(0, \frac13)$}] at (axis cs:0,2) {};
    \node[label=above:{$(\frac23, \frac13)$}] at (axis cs:4,2) {};
    
    \node at (axis cs:1.5,1) {\Large $Z$};
    
    \end{axis}
    \end{tikzpicture}
  \caption{The feasible region of distributions is $R = \{(p, \frac23-p, \frac13 - q, q): (p, q) \in Z\}$, where $Z$ is the polytope above.}
    \label{fig:region}
\end{figure}

\begin{figure}[htbp]
\centering
\begin{tikzpicture}[scale=0.3]
\foreach \a/\b in {0/0,0/-12,14/0,14/-12}{
	\path[-] (\a-3,\b+0.5) edge[thick] node {} (\a-3,\b+2);
}
\foreach \a/\b in {0/-4.5,0/-16.5,14/0,14/-12}{
	\path[-] (\a+3,\b+0.5) edge[thick] node {} (\a+3,\b+2);
}
\foreach \a/\b in {0/-4.5,0/-12,14/-4.5,14/-12}{
	\path[-] (\a+-3,\b-0.5) edge[thick] node {} (\a-3,\b-2);
}
\foreach \a/\b in {0/0,0/-16.5,14/-4.5,14/-12}{
	\path[-] (\a+3,\b-0.5) edge[thick] node {} (\a+3,\b-2);
}
\foreach \a/\b in {0/0,0/-4.5,0/-12,0/-16.5,14/0,14/-4.5,14/-12,14/-16.5}{
\draw [black,line width=1.2pt, dashed] (\a,\b) ellipse (4.2 and 1.7);
\foreach \c/\d in {\a-3/\b,\a/\b,\a+3/\b}{
\draw[bleudefrance,fill=blizzardblue,thick] (\c,\d) circle (24pt); 
}
}
\foreach \a/\b/\c/\d in {0/0/1/A,0/-12/3/C,14/0/2/B,14/-12/4/D}{
\node [draw=none] at (\a-4,\b+3) () {$S_\c$};
\node[state,minimum size=100pt] (\d) at (\a,\b-2.25) {};
}
\end{tikzpicture}
    \caption{The patterns and how they map to states of a Markov chain. The states are unchanged regardless of the number of children: they are defined only with respect to which of the edges are used by the tour. Note that we ignore doubled edges.}
    \label{fig:patterns}
\end{figure}
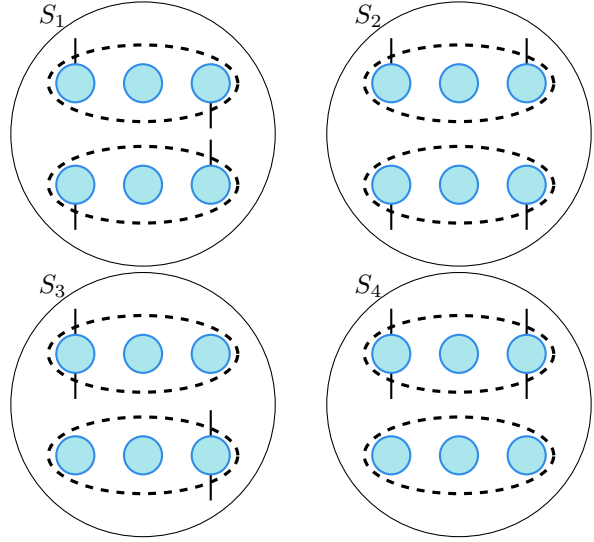

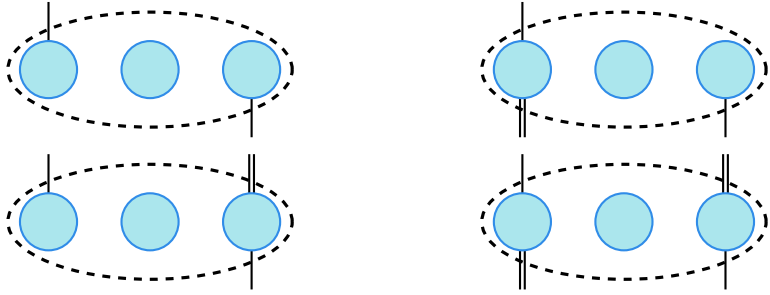
\begin{figure}
    \centering
    \begin{tikzpicture}[scale=0.45]
 \foreach \a/\b in {0/0,0/-4.5,14/0,14/-4.5}{
  	\path[-] (\a-3,\b+0.5) edge[thick] node {} (\a-3,\b+2);
 }
 \foreach \a/\b in {0/-4.5,14/-4.5}{
 	\path[-] (\a+3,\b+0.5) edge[thick,,xshift=-0.07cm] node {} (\a+3,\b+2);
 	\path[-] (\a+3,\b+0.5) edge[thick,xshift=0.07cm] node {} (\a+3,\b+2);
 }
 \foreach \a/\b in {14/0,14/-4.5}{
 	\path[-] (\a+-3,\b-0.5) edge[thick,xshift=-0.07cm] node {} (\a-3,\b-2);
 	\path[-] (\a+-3,\b-0.5) edge[thick,xshift=0.07cm] node {} (\a-3,\b-2);
 }
 \foreach \a/\b in {0/0,0/-4.5,14/0,14/-4.5}{
 	\path[-] (\a+3,\b-0.5) edge[thick] node {} (\a+3,\b-2);
 }
\foreach \a/\b in {0/0,0/-4.5,14/0,14/-4.5}{
\draw [black,line width=1.2pt, dashed] (\a,\b) ellipse (4.2 and 1.7);
\foreach \c/\d in {\a-3/\b,\a/\b,\a+3/\b}{
\draw[bleudefrance,fill=blizzardblue,thick] (\c,\d) circle (24pt); 
}
}
\end{tikzpicture}
    \caption{In our illustrations of the patterns entering a given cycle cut, any edge that is not present may either be unused or doubled. Therefore, all four of the given edge configurations are represented by the upper left most state, $S_1$.}
    \label{fig:doubled}
\end{figure}

To describe the transitions of the Markov chain, we give (randomized) rules that dictate, for a cycle cut $C$ and a pattern entering it, how to connect up its children. These rules depend on whether $C$ is even or odd. The final form of the Markov chains is illustrated in \Cref{fig:general_mcs}.\footnote{In the figure, if there is a variable on an arc, it means that any transition probability in the range of that variable is possible. For example, in $\pe$, we can transition from $S_2$ to $S_1$ with probability $z$ for any $z \in [0,1]$; the transition from $S_2$ to $S_3$ then happens with probability $1 - z$.} The meaning of taking one transition is as follows. Suppose the distribution of patterns entering $C$ is $(p_1, p_2, p_3, p_4)$, and suppose $(q_1, q_2, q_3, q_4)$ is the resulting distribution after one transition of a Markov chain. What this means is that for each child of $C$, the distribution of patterns entering it will be \textbf{either} $(q_1, q_2, q_3, q_4)$ or $(q_2, q_1, q_3, q_4)$ depending on if the child is straight or twisted, respectively (see \Cref{def:intersectingcyclecut} and \cref{fig:caterpillar}). In particular, it can be shown that if $(q_1,q_2,q_3,q_4)$ is the distribution induced on a child which is a straight cycle cut, then $(q_2,q_1,q_3,q_4)$ would be the distribution induced on a child which is a twisted cycle cut. Thus, it is sufficient to check that: i) the distributions induced on straight children lie in the feasible region and ii) if $(q_1,q_2,q_3,q_4)$ is a distribution induced on straight children, then $(q_2,q_1,q_3,q_4)$ is also in the feasible region. This corresponds to the set of distributions induced on the children being symmetric under this transformation.\footnote{Note that the feasible region itself is \emph{not} symmetric under this transformation. The distribution induced on the children is thus a symmetric subset of the feasible region.}


We ensure that in all cases, each edge in $E^{\rightarrow}(C)$ is used $\frac12$, $\frac12$, $1$, $1$ times in expectation if the pattern entering $C$ belongs to state 1, 2, 3, 4, respectively. (We do not explicitly prove this, but it is straightforward to check when we give the rules from connecting the children.) \red{Therefore, if $(p_1, p_2, p_3, p_4)$ are the probabilities that we are in states 1, 2, 3, 4 respectively, then each edge in $E^{\rightarrow}(C)$ is used exactly
$$
\frac12 p_1 + \frac12 p_2 + p_3 + p_4 = 1 - \frac12(p_1 + p_2)
$$
of the time in expectation. Thus, requiring that each edge be used at most $\frac23 = \frac43 x_e$ of the time is equivalent to requiring that $p_1 + p_2 \geq \frac23$. Note that if $\pv \in R$, then $p_1 + p_2 = \frac23$ (i.e. each edge is used \emph{exactly} $\frac23$ of the time).}



\red{
In Section \ref{sec:general_chains}, we illustrate the rules for connecting the children, and show why this leads to the Markov chains in \Cref{fig:general_mcs}. Then in \Cref{sec:algo}, we give a specific example of how to maintain feasible distributions on all the cuts in the hierarchy by choosing appropriate transition probabilities on the Markov chains in \Cref{fig:general_mcs}. This already gives a $\frac43$-approximation algorithm for half-integral cycle cut TSP, which we describe in Algorithm \ref{alg:main}. Finally in \Cref{sec:region_thm}, we complete the picture by proving that $R$ (as given in \Cref{def:region}), is the maximal feasible region of distributions achievable through these chains. 
}

\subsection{The Markov Chains}
\label{sec:general_chains}
We describe the rules to connect the child cuts (given the edges entering the parent), which will allow us to transition according to the Markov chains depicted in \Cref{fig:general_mcs}.\footnote{In the figure, if there is a variable on an arc, it means that any transition probability in the range of that variable is possible. For example, in $\pe$, we can transition from $S_2$ to $S_1$ with probability $z$ for any $z \in [0,1]$; the transition from $S_2$ to $S_3$ then happens with probability $1 - z$.}

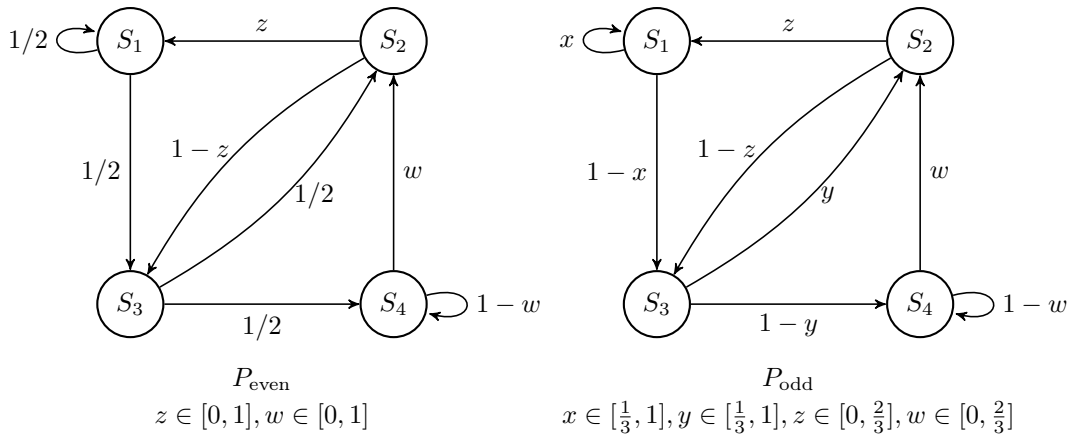
\begin{figure}[ht]
    \centering
        \begin{tikzpicture}[->, >=stealth', auto, semithick, node distance=3.5cm]
\tikzstyle{every state}=[fill=white,draw=black,thick,text=black,scale=1]

\node[state]    (A2)   {$S_1$};
\node[state]    (B2)[right of=A2]   {$S_2$};
\node[state]    (C2)[below of=A2]   {$S_3$};
\node[state]    (D2)[below of=B2]   {$S_4$};
\path
(A2) edge[loop left]     		node{$1/2$}         (A2)
    edge[left]					node{$1/2$}      	(C2)
(B2) edge[above]         		node{$z$}           (A2)
edge[left,bend right=15]         		node{$1-z$}           (C2)
(C2) edge[right,bend right=15]  				node{$1/2$}         (B2)
	edge[below]         			node{$1/2$}         (D2)
(D2) edge[right,loop right]		node{$1-w$}     	(D2)
	edge[right]		    		node{$w$}     	(B2);
\node [draw=none] at (1.75,-5) () {$z \in [0,1], w \in [0,1]$};

\node[state]    (A1)[right of=B2]               {$S_1$};
\node[state]    (B1)[right of=A1]   {$S_2$};
\node[state]    (C1)[below of=A1]   {$S_3$};
\node[state]    (D1)[below of=B1]   {$S_4$};
\path
(A1) edge[loop left]     		node{$x$}         (A1)
    edge[left]					node{$1-x$}      	(C1)
(B1) edge[above]         		node{$z$}         (A1)
	edge[bend right=15,left]    node{$1-z$}         (C1)
(C1) edge[bend right=15,right]  node{$y$}         (B1)
	edge[below]         			node{$1-y$}         (D1)
(D1) edge[right,loop right]		node{$1-w$}     	(D1)
	edge[right]		    		node{$w$}     	(B1);
\node [draw=none] at (1.75,-4.5) () {$P_{\text{even}}$};
\node [draw=none] at (8.75,-4.5) () {$P_{\text{odd}}$};
\node [draw=none] at (8.75,-5) () {$x \in [\frac13,1], y \in [\frac13,1], z \in [0,\frac23], w \in [0, \frac23]$};
\end{tikzpicture}
    \caption{The variables on the arcs indicate that one can feasibly transition according any probability in the range. \bedit{The ranges are given for the cases where the number of children is 2 (in the even case), and 3 (in the odd case), since these are the most restrictive. In general, when there are $k$ children, the ranges are a superset of those given here, and depend on $k$.}}
    \label{fig:general_mcs}
\end{figure}

\begin{prop}
\label{prop:mcs}
For any cycle cut $C \in {\cal H}$ and any distribution of patterns entering $C$, there is a way to connect its children so that the induced distribution on each child is given by 1) applying the corresponding Markov chain in \Cref{fig:general_mcs}, and then 2) swapping the first two coordinates if the child is twisted.
\end{prop}

\begin{proof}
To show that we can always feasibly transition according to the Markov chains in \Cref{fig:general_mcs}, we give rules for connecting the children that result in these transitions. Consider a cycle cut $C$. We consider two cases, depending on whether $C$ is even or odd.

\textbf{Case 1: $C$ is even.} 
\red{We consider the states one by one, and argue that the transitions depicted in \Cref{fig:general_mcs} are achievable.
\begin{enumerate}
    \item \textbf{State 1.} The transitions out of state 1 are depicted in \Cref{fig:smooth_even_s1_general}. For each pair of edges inside the cycle cut, we pick one out of the two uniformly at random. This has the effect of transitioning each child to state 3 with probability $\frac12$, and to state 1 with probability $\frac12$.
    \begin{figure}[!htbp]
        \centering

\begin{tikzpicture}[scale=0.6]
\foreach \a/\b in {0/0,2/0,4/0,12/0,14/0,16/0}{
	\path[-] (\a-2.4,\b+0.5) edge[thick,dotted,black] node {} (\a-1.6,\b+0.5);
	\path[-] (\a-2.4,\b-0.5) edge[thick,dotted,black] node {} (\a-1.6,\b-0.5);
}

\path[-] (-3,0.5) edge[thick,black] node {} (-3,2);
\path[-] (3,-0.5) edge[thick,black] node {} (3,-2);
\path[-] (9,-0.5) edge[thick,black] node {} (9,-2);
\path[-] (15,0.5) edge[thick,black] node {} (15,2);

\foreach \a/\b in {0/0,12/0}{
\draw [black,line width=1.2pt, dashed] (\a,\b) ellipse (4.2 and 1.7);
\foreach \c/\d in {\a-3/\b,\a-1/\b,\a+1/\b,\a+3/\b}{
\draw[bleudefrance,fill=blizzardblue,thick] (\c,\d) circle (20pt); 
}
}

\end{tikzpicture}

        \caption{Transition for state 1 in the even case.}
        \label{fig:smooth_even_s1_general}
    \end{figure}
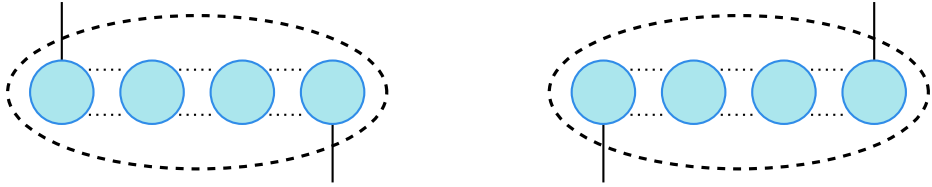
    \item \textbf{State 2.} The way we transition out of this state is depicted in \Cref{fig:smooth_even_s2_general}. With probability $\alpha \in [0,1]$, we alternate taking the top and bottom edges of each pair of edges such that each child transitions back to state 1 with probability 1. (Note that this rule \emph{maximizes} the probability of transitioning back to state 1.) Otherwise, with probability $1-\alpha$, we make all the children transition to state 3, by always picking the top edge of each pair or the bottom edge of each pair. Thus, the net transition probabilities out of state 2 can be made to be $(\alpha, 0, (1-\alpha), 0)$, for any $\alpha \in [0,1]$.
    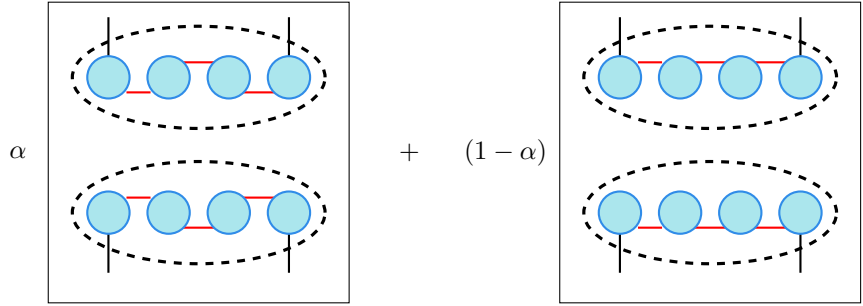
\begin{figure}[!htbp]
        \centering
            \begin{tikzpicture}[scale=0.4]
\path[-] (-2.4,-0.5) edge[thick,red] node {} (-1.6,-0.5);
\path[-] (-1,0.5) edge[thick,red] node {} (1,0.5);
\path[-] (1,-0.5) edge[thick,red] node {} (3,-0.5);

\path[-] (-2.4,-4) edge[thick,red] node {} (-1.6,-4);
\path[-] (-1,-5) edge[thick,red] node {} (1,-5);
\path[-] (1,-4) edge[thick,red] node {} (3,-4);

\path[-] (14.6,0.5) edge[thick,red] node {} (15.4,0.5);
\path[-] (16,0.5) edge[thick,red] node {} (18,0.5);
\path[-] (18,0.5) edge[thick,red] node {} (20,0.5);

\path[-] (14.6,-5) edge[thick,red] node {} (15.4,-5);
\path[-] (16,-5) edge[thick,red] node {} (18,-5);
\path[-] (18,-5) edge[thick,red] node {} (20,-5);

\foreach \a/\b in {0/0,17/0}{
	\path[-] (\a-3,\b+0.5) edge[thick,black] node {} (\a-3,\b+2);
	\path[-] (\a+3,\b+0.5) edge[thick,black] node {} (\a+3,\b+2);
}

\foreach \a/\b in {0/-4.5,17/-4.5}{
    \path[-] (\a-3,\b-0.5) edge[thick,black] node {} (\a-3,\b-2);
	\path[-] (\a+3,\b-0.5) edge[thick,black] node {} (\a+3,\b-2);
}
\foreach \a/\b in {0/0,0/-4.5,17/0,17/-4.5}{
\draw [black,line width=1.2pt, dashed] (\a,\b) ellipse (4.2 and 1.7);
\foreach \c/\d in {\a-3/\b,\a-1/\b,\a+1/\b,\a+3/\b}{
\draw[bleudefrance,fill=blizzardblue,thick] (\c,\d) circle (20pt); 
}
}
\draw (-5,2.5) rectangle (5,-7.5);
\draw (12,2.5) rectangle (22,-7.5);
\node [draw=none] at (-6,-2.5) () {$\alpha$};
\node [draw=none] at (7,-2.5) () {$+$};
\node [draw=none] at (10.2,-2.5) () {$(1-\alpha)$};
\end{tikzpicture}
        
        \caption{Transition for state 2 in the even case.}
        \label{fig:smooth_even_s2_general}
    \end{figure}
    \item \textbf{State 3.} The transitions out of this state are depicted in \Cref{fig:smooth_even_s3_general}.     To connect up the children in this state, we consider each pair of edges in $E^{\rightarrow}(C)$ independently. Let $e,f$ be such a pair. Then with probability $\frac12$, we use both $e$ and $f$ (one copy each), as is illustrated by the solid orange edges. Otherwise, we either double $e$ or double $f$, with equal probability, shown using the dotted black edges. The net effect is that each child transitions to state 2 with probability $\frac12$, and to state 4 with probability $\frac12$.

    \begin{figure}[!htbp]
        \centering

        \begin{tikzpicture}[scale=0.6]
\foreach \a/\b in {0/0,2/0,4/0,12/0,14/0,16/0}{
	\path[-] (\a-2.4,\b+0.53) edge[thick,dotted,black] node {} (\a-1.6,\b+0.53);
	\path[-] (\a-2.4,\b+0.4) edge[thick,dotted,black] node {} (\a-1.6,\b+0.4);
	\path[-] (\a-2.4,\b-0.53) edge[thick,dotted,black] node {} (\a-1.6,\b-0.53);
	\path[-] (\a-2.4,\b-0.4) edge[thick,dotted,black] node {} (\a-1.6,\b-0.4);
	\path[-] (\a-3,\b+0.6) edge[thick,orange,bend left=10] node {} (\a-1,\b+0.6);	
	\path[-] (\a-3,\b-0.6) edge[thick,orange,bend right=10] node {} (\a-1,\b-0.6);
}

\path[-] (-3,0.5) edge[thick,black] node {} (-3,2);
\path[-] (-3,-0.5) edge[thick,black] node {} (-3,-2);
\path[-] (15,-0.5) edge[thick,black] node {} (15,-2);
\path[-] (15,0.5) edge[thick,black] node {} (15,2);

\foreach \a/\b in {0/0,12/0}{
\draw [black,line width=1.2pt, dashed] (\a,\b) ellipse (4.2 and 1.7);
\foreach \c/\d in {\a-3/\b,\a-1/\b,\a+1/\b,\a+3/\b}{
\draw[bleudefrance,fill=blizzardblue,thick] (\c,\d) circle (20pt); 
}
}

\end{tikzpicture}
        
        \caption{Transition for state 3 in the even case.}
        \label{fig:smooth_even_s3_general}
    \end{figure}
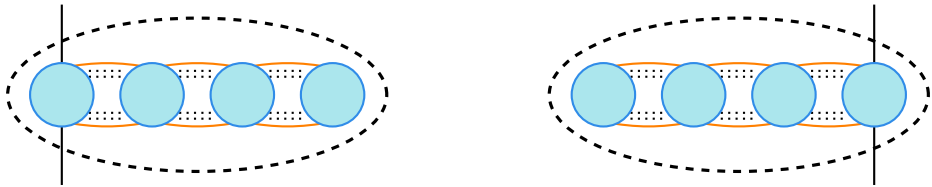
    \item \textbf{State 4.} The transitions out of this state are depicted in \Cref{fig:smooth_even_s4_general}.  With probability $\alpha$, we make all children transition to state 2 with probability 1. To do this, first suppose $C$ has all 4 single edges entering it (the top picture in the left box). In this case, we consider the pairs of edges in $E^{\rightarrow}(C)$ from left to right, and alternate 1) doubling one of the two edges with equal probability (shown by the dotted black edges), and 2) using both edges (shown by the solid black edges). Because $C$ is even, the rightmost pair of edges ends up falling in case 1) of the alternating rule, and so all children transition to state 2. The case where all the edges entering $C$ are used an even number of times (the bottom picture in the left box) is quite similar, except we begin the alternating rule by using both edges. 
    
    On the other hand, with probability $1-\alpha$, we transition back to state 4 with probability 1. This is accomplished by using each pair of edges in the top case of state 4, and by doubling one edge from each pair uniformly at random in the bottom case of state 4. The net transition probabilities are then $(0, \alpha, 0, 1 -\alpha)$, where $\alpha$ can be any number from 0 to 1.

    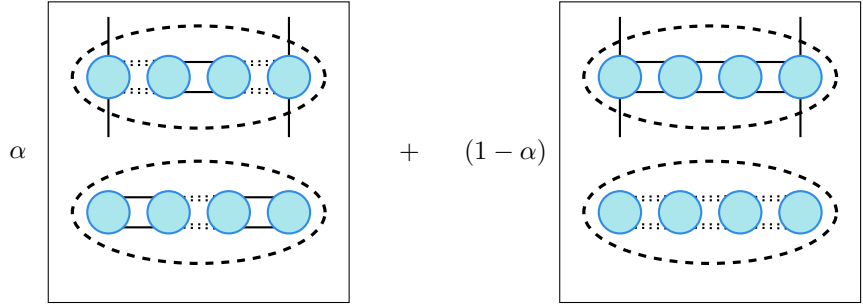
\begin{figure}[!htbp]
    \centering
\begin{tikzpicture}[scale=0.4]
\foreach \a/\b in {-3/0,1/0,-1/-4.5,14/-4.5,16/-4.5,18/-4.5}{
	\path[-] (\a,\b+0.5) edge[thick,dotted] node {} (\a+2,\b+0.55);
	\path[-] (\a,\b+0.4) edge[thick,dotted] node {} (\a+2,\b+0.4);
	\path[-] (\a,\b-0.5) edge[thick,dotted] node {} (\a+2,\b-0.55);
	\path[-] (\a,\b-0.4) edge[thick,dotted] node {} (\a+2,\b-0.4);
}

\foreach \a/\b in {-1/0,-3/-4.5,1/-4.5,14/0,16/0,18/0} {
	\path[-] (\a,\b+0.5) edge[thick] node {} (\a+2,\b+0.5);
	\path[-] (\a,\b-0.5) edge[thick] node {} (\a+2,\b-0.5);
}

\foreach \a/\b in {0/0,17/0}{
	\path[-] (\a-3,\b+0.5) edge[thick,black] node {} (\a-3,\b+2);
}
\foreach \a/\b in {0/0,17/0}{
	\path[-] (\a+3,\b+0.5) edge[thick,black] node {} (\a+3,\b+2);
}
\foreach \a/\b in {0/0,17/0}{
	\path[-] (\a+-3,\b-0.5) edge[thick,black] node {} (\a-3,\b-2);
}
\foreach \a/\b in {0/0,17/0}{
	\path[-] (\a+3,\b-0.5) edge[thick,black] node {} (\a+3,\b-2);
}
\foreach \a/\b in {0/0,0/-4.5,17/0,17/-4.5}{
\draw [black,line width=1.2pt, dashed] (\a,\b) ellipse (4.2 and 1.7);
\foreach \c/\d in {\a-3/\b,\a-1/\b,\a+1/\b,\a+3/\b}{
\draw[bleudefrance,fill=blizzardblue,thick] (\c,\d) circle (20pt); 
}
}
\draw (-5,2.5) rectangle (5,-7.5);
\draw (12,2.5) rectangle (22,-7.5);
\node [draw=none] at (-6,-2.5) () {$\alpha$};
\node [draw=none] at (7,-2.5) () {$+$};
\node [draw=none] at (10.2,-2.5) () {$(1-\alpha)$};
\end{tikzpicture}
    \caption{Transition for state 4 in the even case.}
    \label{fig:smooth_even_s4_general}
    \end{figure}

 \end{enumerate}}

This shows that the general form of the even chains, as depicted in \Cref{fig:general_mcs}, are achievable.\footnote{\bedit{Actually, note that slightly more general transitions out of states 2 and 3 are possible as a function of $k$, the number of children. For example, one can show (similarly to the odd case) there are rules for connecting the children that allow the transition from state 3 to state 2 to be any value in the range $[\frac{1}{k}, \frac{k-1}{k}]$. However, the transitions are most restricted when $k=2$, which results in the Markov chains we presented in \Cref{fig:general_mcs}. (e.g. Note that $\frac1k = \frac{k-1}{k} = \frac12$ if $k=2$.)} }


\textbf{Case 2: $C$ is odd.} To show that we can feasibly transition according to the Markov chain $\po$ in \Cref{fig:general_mcs}, we consider each state one by one.
\begin{enumerate}
    \item \textbf{State 1.} The way we transition out of this state is depicted in \Cref{fig:smooth_odd_s1_general}. With probability $\alpha \in [0,1]$, we alternate taking the top and bottom edges of each pair of edges such that each child transitions back to state 1 with probability 1. (Note that this rule \emph{maximizes} the probability of transitioning back to state 1.) Otherwise, with probability $1-\alpha$, we choose one of the children uniformly at random to transition to state 1 (the rightmost child in the figure), and make all other children transition to state 3. Note that once we have chosen which child to transition to state 1, there is a unique choice of edges that makes that child transition to state 1 and all other children transition to state 3. The net effect is that the children  transition to state 1 with probability $\frac1k$, and to state 3 with probability $1 - \frac1k$, where $k$ is the number of children. (This rule \emph{minimizes} the probability of transitioning back to state 1.) Thus, the net transition probabilities out of state 1 are $(\alpha + \frac{1-\alpha}{k}, 0, (1-\alpha)\cdot \frac{k-1}{k}, 0)$. As $\alpha$ ranges from 0 to 1, the probability of transitioning back to state 1 ranges from $\frac1k$ to $1$. Since this range is most restricted when $k=3$, we conclude that it is always feasible to transition out of state 1 according to the probabilities $(x, 0, 1-x, 0)$, for any $x \in [\frac13, 1]$. 
    
    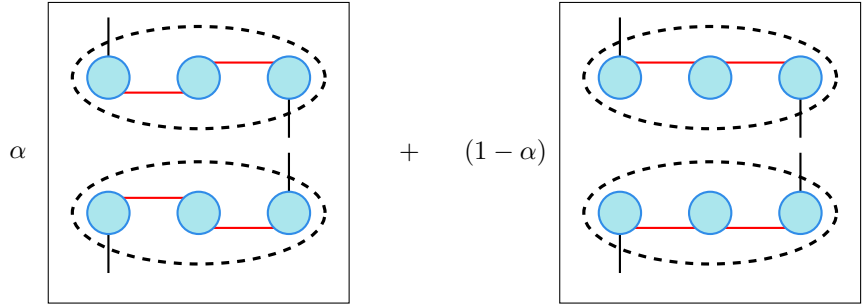
\begin{figure}[!htbp]
    \centering
    \begin{tikzpicture}[scale=0.4]

\path[-] (-3,-0.5) edge[thick,red] node {} (0,-0.5);
\path[-] (0,0.5) edge[thick,red] node {} (3,0.5);

\path[-] (-3,-4) edge[thick,red] node {} (0,-4);
\path[-] (0,-5) edge[thick,red] node {} (3,-5);

\foreach \a/\b in {14/0,17/0}{
	\path[-] (\a,\b+0.5) edge[thick, red] node {} (\a+3,\b+0.5);
}

\foreach \a/\b in {14/-4.5,17/-4.5}{
	\path[-] (\a,\b-0.5) edge[thick, red] node {} (\a+3,\b-0.5);
}

\foreach \a/\b in {0/0,17/0}{
	\path[-] (\a-3,\b+0.5) edge[thick,black] node {} (\a-3,\b+2);
}
\foreach \a/\b in {0/-4.5,17/-4.5}{
	\path[-] (\a+3,\b+0.5) edge[thick,black] node {} (\a+3,\b+2);
}
\foreach \a/\b in {0/-4.5,17/-4.5}{
	\path[-] (\a+-3,\b-0.5) edge[thick,black] node {} (\a-3,\b-2);
}
\foreach \a/\b in {0/0,17/0}{
	\path[-] (\a+3,\b-0.5) edge[thick,black] node {} (\a+3,\b-2);
}
\foreach \a/\b in {0/0,0/-4.5,17/0,17/-4.5}{
\draw [black,line width=1.2pt, dashed] (\a,\b) ellipse (4.2 and 1.7);
\foreach \c/\d in {\a-3/\b,\a/\b,\a+3/\b}{
\draw[bleudefrance,fill=blizzardblue,thick] (\c,\d) circle (20pt); 
}
}
\draw (-5,2.5) rectangle (5,-7.5);
\draw (12,2.5) rectangle (22,-7.5);
\node [draw=none] at (-6,-2.5) () {$\alpha$};
\node [draw=none] at (7,-2.5) () {$+$};
\node [draw=none] at (10.2,-2.5) () {$(1-\alpha)$};
\end{tikzpicture}
    \caption{Transitions out of  state 1 in the $\po$ chain. On the left, each child transitions back to state 1. On the right, we pick one of the children uniformly at random to transition back to state 1 (we visualize this to be the rightmost child in the picture), and the remaining children transition to state 3.}
    \label{fig:smooth_odd_s1_general}
\end{figure}
    
    \item \textbf{State 2.} The way we transition out of state 2 is depicted in \Cref{fig:smooth_odd_s2_general}. With probability $\alpha$, the net transition probabilities out of state 2 are $(\frac{k-1}{k}, 0, \frac1k, 0)$, where $k$ is the number of children. This is accomplished by choosing one of the children uniformly at random to transition to state 3, and the having the remaining children transition to state 1.  In more detail, suppose the children are $a_1, \ldots, a_k$ from left to right, and suppose we chose $a_i$ to transition to state 3. In the case where the two edges enter $C$ from the top (the top left picture in \Cref{fig:smooth_odd_s2_general}), we go through $a_1, a_2, \ldots, a_i$, alternatingly using the bottom edge from $a_1$ to $a_2$, the top edge from $a_2$ to $a_3$, and so on, until we reach $a_i$. We then go through $a_k, a_{k-1}, \ldots, a_i$, using the bottom edge from $a_k$ to $a_{k-1}$, then the top edge from $a_{k-1}$ to $a_{k-2}$, and so on, until we reach $a_i$.  Since $C$ is odd, $a_i$ will end up either having two edges incident to it from the top or two edges incident to it from the bottom. Thus, $a_i$ transitions to state 3, and all the children except for $a_i$ transition to state 1. The case where the two edges enter $C$ from the bottom (the bottom left picture in \Cref{fig:smooth_odd_s2_general}) is the mirror image of this. Finally, since the child $a_i$ which transitions to state 3 is chosen uniformly at random, the net transition probabilities are  $(\frac{k-1}{k}, 0, \frac1k, 0)$.
    
    On the other hand, with probability $1-\alpha$ we always transition to state 3, by either always taking the top edge of each pair or the bottom edge of each pair, depending on if the two edges incident to $C$ enter from the top or the bottom, respectively. The overall transition probabilities are therefore $(\alpha \cdot \frac{k-1}{k}, 0, \frac{\alpha}{k} + (1-\alpha), 0)$. Since the range of transition probabilities is most constrained when $k=3$; we conclude it is feasible to transition out of state 2 with probabilities $(z, 0, 1-z, 0)$ for any $z \in [0, \frac23]$. 
    
        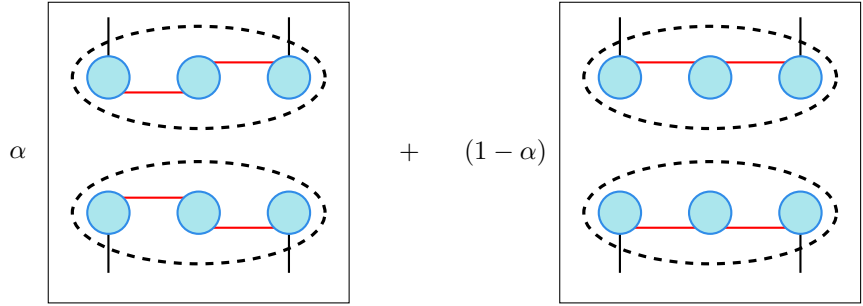
\begin{figure}[!htbp]
    \centering
    \begin{tikzpicture}[scale=0.4]

\path[-] (-3,-0.5) edge[thick,red] node {} (0,-0.5);
\path[-] (0,0.5) edge[thick,red] node {} (3,0.5);

\path[-] (-3,-4) edge[thick,red] node {} (0,-4);
\path[-] (0,-5) edge[thick,red] node {} (3,-5);

\foreach \a/\b in {14/0,17/0}{
	\path[-] (\a,\b+0.5) edge[thick, red] node {} (\a+3,\b+0.5);
}

\foreach \a/\b in {14/-4.5,17/-4.5}{
	\path[-] (\a,\b-0.5) edge[thick, red] node {} (\a+3,\b-0.5);
}

\foreach \a/\b in {0/0,17/0}{
	\path[-] (\a-3,\b+0.5) edge[thick,black] node {} (\a-3,\b+2);
	\path[-] (\a+3,\b+0.5) edge[thick,black] node {} (\a+3,\b+2);
}
\foreach \a/\b in {0/-4.5,17/-4.5}{
	\path[-] (\a-3,\b-0.5) edge[thick,black] node {} (\a-3,\b-2);
	\path[-] (\a+3,\b-0.5) edge[thick,black] node {} (\a+3,\b-2);
}

\foreach \a/\b in {0/0,0/-4.5,17/0,17/-4.5}{
\draw [black,line width=1.2pt, dashed] (\a,\b) ellipse (4.2 and 1.7);
\foreach \c/\d in {\a-3/\b,\a/\b,\a+3/\b}{
\draw[bleudefrance,fill=blizzardblue,thick] (\c,\d) circle (20pt); 
}
}
\draw (-5,2.5) rectangle (5,-7.5);
\draw (12,2.5) rectangle (22,-7.5);
\node [draw=none] at (-6,-2.5) () {$\alpha$};
\node [draw=none] at (7,-2.5) () {$+$};
\node [draw=none] at (10.2,-2.5) () {$(1-\alpha)$};
\end{tikzpicture}
    \caption{Transitions out of state 2 in the $\po$ chain. On the left, we pick one of the children uniformly at random to transition to state 3 (we visualize this to be the rightmost child in the picture), and the remaining children transition to state 1. On the right, each child transitions to state 3.}
    \label{fig:smooth_odd_s2_general}
\end{figure}

    \item \textbf{State 3.} The transition out of state 3 is depicted in \Cref{fig:smooth_odd_s3_general}. With probability $\alpha$, every child transitions to state 2. To do this, we start at the child with two edges entering it (i.e. the leftmost child in the top left picture of \Cref{fig:smooth_odd_s3_general}, and the rightmost child in the bottom left picture of \Cref{fig:smooth_odd_s3_general}), and for each pair of edges in $E^{\rightarrow}{(C)}$, we alternate 1) doubling one of the two edges with equal probability (shown by the dotted black edges), and 2) using both edges once (shown by the solid black edges). 
    
    On the other hand, with probability $1-\alpha$, we choose one child uniformly at random to transition to state 2, and the remaining children transition to state 4. The way we accomplish this is as follows: Suppose the children are $a_1, \ldots, a_k$ from left to right, and suppose we chose $a_i$ to transition to state 2. Then in the case where the two edges enter $C$ from the left (the top right picture in \Cref{fig:smooth_odd_s3_general}), we use both edges in each pair going from $a_1, a_2,\ldots, $ all the way to $a_i$. Then from $a_i$ to $a_k$, we double one edge uniformly at random between each pair. The case where the two edges enter $C$ from the left (the bottom right picture in \Cref{fig:smooth_odd_s3_general}) is the mirror image of this; we double one edge uniformly at random between each pair between $a_1$ and $a_i$, and then use both edges in each pair going from $a_i$ to $a_k$. Since the child which transitions to state 2 is chosen uniformly at random, the transition probabilities are then $(0, \frac{1}{k}, 0, \frac{k-1}{k})$. Taking the convex combination of this with the earlier rule which transitions to state 2 deterministically, the net transition probabilities are then $(0, \alpha + \frac{1-\alpha}{k}, 0, (1-\alpha) \cdot \frac{k-1}{k})$. As $\alpha$ ranges from 0 to 1, the transition probability to state 2 ranges from $\frac1k$ to 1. Since this range is most restricted when $k=3$, we conclude that it is always feasible to transition with probabilities $(0, y, 0, 1-y)$, for any $y \in [\frac13, 1]$.

    \begin{figure}[!htbp]
    \centering
        \begin{tikzpicture}[scale=0.4]

\foreach \a/\b in {-3/0,0/-4.5}{
	\path[-] (\a,\b+0.5) edge[thick,dotted] node {} (\a+3,\b+0.5);
	\path[-] (\a,\b+0.4) edge[thick,dotted] node {} (\a+3,\b+0.4);	
	\path[-] (\a,\b-0.5) edge[thick,dotted] node {} (\a+3,\b-0.5);
	\path[-] (\a,\b-0.4) edge[thick,dotted] node {} (\a+3,\b-0.4);
}

\foreach \a/\b in {0/0,-3/-4.5}{
	\path[-] (\a,\b+0.5) edge[thick] node {} (\a+3,\b+0.5);
	\path[-] (\a,\b-0.5) edge[thick] node {} (\a+3,\b-0.5);
}

\foreach \a/\b in {14/0,17/0}{
	\path[-] (\a,\b+0.6) edge[thick,orange,bend left=10] node {} (\a+3,\b+0.6);	
	\path[-] (\a,\b-0.6) edge[thick,orange,bend right=10] node {} (\a+3,\b-0.6);
}

\foreach \a/\b in {14/-4.5,17/-4.5}{
	\path[-] (\a,\b+0.5) edge[thick,dotted] node {} (\a+3,\b+0.5);
	\path[-] (\a,\b+0.4) edge[thick,dotted] node {} (\a+3,\b+0.4);	
	\path[-] (\a,\b-0.5) edge[thick,dotted] node {} (\a+3,\b-0.5);
	\path[-] (\a,\b-0.4) edge[thick,dotted] node {} (\a+3,\b-0.4);
}

\foreach \a/\b in {0/0,17/0}{
	\path[-] (\a-3,\b+0.5) edge[thick,black] node {} (\a-3,\b+2);
	}
	
\foreach \a/\b in {0/-4.5,17/-4.5}{
	\path[-] (\a+3,\b+0.5) edge[thick,black] node {} (\a+3,\b+2);
}
\foreach \a/\b in {0/0,17/0}{
	\path[-] (\a+-3,\b-0.5) edge[thick,black] node {} (\a-3,\b-2);
}
\foreach \a/\b in {0/-4.5,17/-4.5}{
	\path[-] (\a+3,\b-0.5) edge[thick,black] node {} (\a+3,\b-2);
}
\foreach \a/\b in {0/0,0/-4.5,17/0,17/-4.5}{
\draw [black,line width=1.2pt, dashed] (\a,\b) ellipse (4.2 and 1.7);
\foreach \c/\d in {\a-3/\b,\a/\b,\a+3/\b}{
\draw[bleudefrance,fill=blizzardblue,thick] (\c,\d) circle (20pt); 
}
}
\draw (-5,2.5) rectangle (5,-7.5);
\draw (12,2.5) rectangle (22,-7.5);
\node [draw=none] at (-6,-2.5) () {$\alpha$};
\node [draw=none] at (7,-2.5) () {$+$};
\node [draw=none] at (10.2,-2.5) () {$(1-\alpha)$};
\end{tikzpicture}
    \caption{Transitions out of state 3 in the general $\po$ chain. On the left, every child transitions to state 2. On the right, we pick one child uniformly at random to transition to state 2 (we visualize this to be the right child in the figure), and the remaining children transition to state 4.}
    \label{fig:smooth_odd_s3_general}
\end{figure}
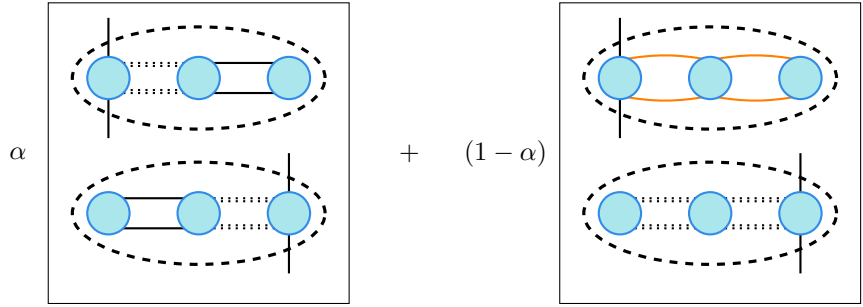
    
    \item \textbf{State 4.} \bedit{The transition out of state 4 is depicted in \Cref{fig:smooth_odd_s4_general}. With probability $\alpha$, we choose one child uniformly at random to transition to state 4, and make the other children transition to state 2. The way we accomplish this is as follows: Suppose the children are $a_1, \ldots, a_k$ from left to right, and suppose $a_i$ is chosen to transition to state 4. Then in the case where $C$ has 4 edges entering it (the top left picture in \Cref{fig:smooth_odd_s4_general}), we go through $a_1, a_2, \ldots, a_i$ from left to right, and alternatingly 1) double an edge from each pair of edges uniformly at random and 2) use both edges once, until we reach $a_i$. We do the same on the other side from $a_k, a_{k-1}, \ldots $, until $a_i$. This will cause every child except for $a_i$ to be in state 2, and (since $k$ is odd), $a_i$ will be in state 4. The case where every edge entering $C$ is used an even number of times (the bottom picture in the left box of \Cref{fig:smooth_odd_s4_general}) is similar, except we begin the alternation by using both edges. The transition probabilities in this case are $(0, \frac{k-1}{k}, 0, \frac{1}{k})$, where $k$ is the number of children.
    
    On the other hand, with probability $1-\alpha$, every child transitions to state 4, by either using both edges from each pair if we are in the top case of state 4, or by doubling one edge from each pair uniformly at random if we are in the bottom case of state 4. Overall, the net transition probabilities are then $(0,\alpha \cdot \frac{k-1}{k}, 0, \frac{\alpha}{k} + 1-\alpha)$. As $\alpha$ ranges from 0 to 1, the transition probability to state 2 ranges from $0$ to $\frac{k-1}{k}$. Since this range is most restricted when $k=3$, we conclude it is always feasible to transition out of state 4 with probabilities $(0, w, 0, 1-w)$ for any $w \in [0, \frac{2}{3}]$. 
    
    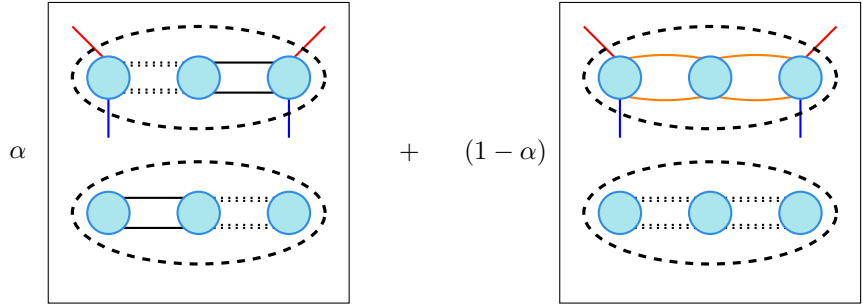
\begin{figure}[!htbp]
    \centering
        \begin{tikzpicture}[scale=0.4]

\foreach \a/\b in {-3/0,0/-4.5}{
	\path[-] (\a,\b+0.5) edge[thick,dotted] node {} (\a+3,\b+0.5);
	\path[-] (\a,\b+0.4) edge[thick,dotted] node {} (\a+3,\b+0.4);	
	\path[-] (\a,\b-0.5) edge[thick,dotted] node {} (\a+3,\b-0.5);
	\path[-] (\a,\b-0.4) edge[thick,dotted] node {} (\a+3,\b-0.4);
}

\foreach \a/\b in {0/0,-3/-4.5}{
	\path[-] (\a,\b+0.5) edge[thick] node {} (\a+3,\b+0.5);
	\path[-] (\a,\b-0.5) edge[thick] node {} (\a+3,\b-0.5);
}

\foreach \a/\b in {14/0,17/0}{
	\path[-] (\a,\b+0.6) edge[thick,orange,bend left=10] node {} (\a+3,\b+0.6);	
	\path[-] (\a,\b-0.6) edge[thick,orange,bend right=10] node {} (\a+3,\b-0.6);
}

\foreach \a/\b in {14/-4.5,17/-4.5}{
	\path[-] (\a,\b+0.5) edge[thick,dotted] node {} (\a+3,\b+0.5);
	\path[-] (\a,\b+0.4) edge[thick,dotted] node {} (\a+3,\b+0.4);	
	\path[-] (\a,\b-0.5) edge[thick,dotted] node {} (\a+3,\b-0.5);
	\path[-] (\a,\b-0.4) edge[thick,dotted] node {} (\a+3,\b-0.4);
}

\foreach \a/\b in {0/0,17/0}{
	\path[-] (\a-3,\b+0.5) edge[thick,red] node {} (\a-4.2,\b+1.7);
	\path[-] (\a+3,\b+0.5) edge[thick,red] node {} (\a+4.2,\b+1.7);
}

\foreach \a/\b in {0/0,17/0}{
	\path[-] (\a+-3,\b-0.5) edge[thick,blue] node {} (\a-3,\b-2);
	\path[-] (\a+3,\b-0.5) edge[thick,blue] node {} (\a+3,\b-2);
}

\foreach \a/\b in {0/0,0/-4.5,17/0,17/-4.5}{
\draw [black,line width=1.2pt, dashed] (\a,\b) ellipse (4.2 and 1.7);
\foreach \c/\d in {\a-3/\b,\a/\b,\a+3/\b}{
\draw[bleudefrance,fill=blizzardblue,thick] (\c,\d) circle (20pt); 
}
}
\draw (-5,2.5) rectangle (5,-7.5);
\draw (12,2.5) rectangle (22,-7.5);
\node [draw=none] at (-6,-2.5) () {$\alpha$};
\node [draw=none] at (7,-2.5) () {$+$};
\node [draw=none] at (10.2,-2.5) () {$(1-\alpha)$};
\end{tikzpicture}
    \caption{Transitions out of state 4 in the general $\po$ chain. On the left, we pick one child uniformly at random to transition to state 4 (we visualize this to be the right child in the figure), and the remaining children transition to state 2. On the right, every child transitions to state 4.}
    \label{fig:smooth_odd_s4_general}
\end{figure}
}
\end{enumerate}

This finishes the proof of why the $\po$ Markov chain is achievable for the ranges of probabilities depicted in \Cref{fig:general_mcs}.


\end{proof}

Note that the rules given in the proof of \Cref{prop:mcs} satisfy the following two invariants: 1) Given that a cycle cut $C$ is in some state, it is equally likely to look like the top picture as the bottom picture of \Cref{fig:patterns}, and 2) Each edge in $E^{\rightarrow}(C)$ is used $\frac12$, $\frac12$, $1$, $1$ times in expectation if $C$ is in state 1, 2, 3, 4, respectively. 

\subsection{Fixed Point and Algorithm}
\label{sec:algo}

\begin{algorithm}[t]
\caption{A randomized $\frac43$-approximation algorithm for half-integral cycle cut TSP.}
\label{alg:main}
\DontPrintSemicolon
\KwIn{A half-integral cycle cut TSP instance $G = (V,E)$ with edge costs $c$.}
\KwOut{An Eulerian tour $T$ with expected cost at most $\frac43$ times that of the Subtour LP.}
\begin{algorithmic}[1]
\STATE Compute $x$, an optimal solution of the Subtour LP.
\STATE Choose any vertex $r \in V$, and compute the hierarchy $\mathcal{H}$ of critical cuts in $V \setminus \{r\}$. 
\IF{$x$ is not half-integral \textbf{or} some cut in $\mathcal{H}$ is not a cycle cut}
\STATE \textbf{Fail} \COMMENT{not a half-integral cycle cut instance}
\ENDIF
\STATE Initialize $T \gets \emptyset$.
\STATE Sample edges entering $V \setminus \{r\}$ according to the distribution $\pv = (\frac49, \frac29, \frac29, \frac19)$. Add these edges to $T$. 
\FOR{each cut $C$ in a depth-first search ordering of $\mathcal{H}$}
\STATE Given the edges in $T$ entering $C$, sample edges connecting the children of $C$ according to the \\ rules described in \Cref{sec:general_chains} using the specific transition probabilities in \Cref{sec:algo}. Add these edges to $T$.
\ENDFOR
\STATE \textbf{Return} $T$.
\end{algorithmic}
\end{algorithm}

We now give the reader some more intuition by giving a specific example of how to maintain distributions in the feasible region $R$ (as defined in \Cref{def:region}), on all the cuts in the hierarchy by choosing appropriate transition probabilities on the Markov chains in \Cref{fig:general_mcs}. This already gives a $\frac43$-approximation algorithm for half-integral cycle cut TSP, which we describe in Algorithm \ref{alg:main}. In \Cref{sec:region_thm}, we will extend the ideas here to show that $R$ is the maximal feasible region achievable through our Markov chains.

Specifically, let $\pv = (\frac49, \frac29, \frac29, \frac19)$ and $\qv = (\frac29, \frac49, \frac29, \frac19)$ (i.e. $\qv$ is $\pv$ with the first two coordinates swapped). It is easy to check that $\pv, \qv \in R$. We now show for any half-integral cycle cut instance, it is possible to make the distribution entering any cycle cut to be either $\pv$ or $\qv$. 

To see this, let $C$ be a cycle cut and suppose $C$ is odd. Set the transition probabilities in $\po$ to be $x = y = z = w = \frac23$. For these probabilities, it is easy to check that $\po \pv = \po \qv = \pv$.\footnote{In fact, it can be checked that for these probabilities, $\po$ maps \emph{every} distribution (whose first two coordinates sum to $\frac23$), to $\pv$.} On the other hand, if $C$ is even, setting $z = w = 1$ in $\pe$ gives $\pe \pv = \pv$, and setting $z = \frac34, w = 1$ gives $\pe \qv = \pv$. Thus, as long as the distribution entering $C$ is $\pv$ or $\qv$, we can make the distribution on each child of $C$ be either $\pv$ (if the child is straight), or $\qv$ (if the child is twisted). Since we have freedom in choosing the distribution on the topmost cycle cut $V \setminus \{r\}$, we can simply set it to be $\pv$, and then following the rules given in \Cref{sec:general_chains} with the above transition probabilities will ensure that the distribution on every cut in the hierarchy is either $\pv$ or $\qv$.

\begin{prop}
\label{prop:algo}
    Algorithm \ref{alg:main} is a $\frac43$-approximation algorithm for half-integral cycle cut instances of the TSP.
\end{prop}
\begin{proof}
    By the above reasoning, Algorithm \ref{alg:main} samples from a distribution of Eulerian tours with the property that the distribution of patterns on each cut in the hierarchy is either $\pv$ or $\qv$. Under $\pv$ and $\qv$, the rules for connecting the children in \Cref{sec:general_chains} guarantee that each edge is used exactly $\frac43 x_e$ of the time in expectation. It remains to verify that the algorithm runs in polynomial time. To begin, the Subtour LP can be solved in polynomial time using the ellipsoid method  since it has an efficient separation oracle \cite{grotschel2012geometric}, which is given by a global min-cut computation (see, e.g. Chapter 3 of \cite{Williamson19}). Moreover, the Subtour LP can be solved very quickly in practice \cite{JohnsonM07}. The hierarchy of critical cuts can also be found efficiently by computing the cactus decomposition of the graph (e.g. \cite{Fleischer99}). Finally, given the hierarchy, sampling the tour just requires going through the cuts in the hierarchy from the top-down (e.g. using a depth-first search), and for each cut following the rules to sample a multiset of edges inside of it. This takes linear time in the size of the graph. 
    
\end{proof}


\subsection{Characterizing the Feasible Region}
\label{sec:region_thm}
We now show that $R$ (as given in \Cref{def:region}), is the maximal feasible region according to the chains in \Cref{fig:general_mcs}. Recall that by "feasible region", we mean that 
1)  If the distribution of patterns entering a cycle cut $C$ belongs to $R$, there is a way to connect up the children of $C$ such that the distribution on each child also belongs to $R$, and 2) for each $\pv \in R$, the corresponding rule for connecting the children of $C$ uses each edge in $E^{\rightarrow}(C)$ at most $\frac23 = \frac43 x_e$ of the time in expectation.
Informally speaking, $R$ is the set of distributions "that guarantee a $\frac43$-approximation all the way down" the hierarchy of cycle cuts. 
\red{In particular, the fact that $R$ is nonempty implies the existence of a $\frac43$-approximation algorithm.}

\begin{remark}
Note that the distribution $(\frac13, \frac13, \frac13, 0)$ lies in $R$. This is the distribution that, among the four edges entering a cycle cut, uses each pair of edges with equal probability (and possibly doubles zero, one or both of the other edges). We find it nice that such a symmetric distribution is feasible.
\end{remark}


Before moving on, note that every distribution in $R$ has a net probability of $\frac23$ to be in states 1 or 2, and a probability of $\frac13$ to be in states 3 or 4. For any cycle cut $C$, since all of our rules for connecting the children use each edge in $E^{\rightarrow}(C)$ $\frac12$ the time in states 1 and 2, and once in expectation in states 3 and 4, every distribution in $R$ automatically uses each edge exactly $\frac23 = \frac43 x_e$ times in expectation. Therefore, checking the feasibility of $R$ boils down to showing that if the distribution of a parent belongs to $R$, then we can make the distribution of the children also belong to $R$. 

We show this in \Cref{thm:region_sufficient}. In other words, $R$ is \emph{sufficient}, in the sense that if the distribution entering a cycle cut belongs to $R$, then it is possible to get a $\frac43$-approximation all the way down the hierarchy using the Markov chains in \Cref{fig:general_mcs}. We complement this by showing in \Cref{thm:region_necessary} that $R$ is \emph{necessary}; if the distribution entering a cycle cut does \textbf{not} belong to $R$, then it is impossible to obtain a $\frac43$-approximation using these Markov chains. 

\begin{theorem}[$R$ is sufficient]
\label{thm:region_sufficient}
If the distribution of patterns entering a cycle cut belongs to $R$, then there are feasible Markov chains (among the ones shown in \Cref{fig:general_mcs}) such that the induced distribution entering each child also belongs to $R$. 
\end{theorem}

\begin{theorem}[$R$ is necessary]
\label{thm:region_necessary}
Suppose the distribution of patterns entering a cycle cut does \textbf{not} belong to $R$. Then it is not possible to obtain a $\frac43$-approximation using the Markov chains in \Cref{fig:general_mcs}. 
\end{theorem}

\red{
\begin{proof}[Proof of \Cref{thm:region_sufficient}]
Let $C$ be any cycle cut in the hierarchy. Suppose the distribution entering $C$ is $(p_1, p_2, p_3, p_4) \in R$. We consider the 2 cases, depending on if $C$ is even or odd. We show that in each case, there is a valid choice of transition probabilities for the corresponding Markov chain (illustrated in \Cref{fig:general_mcs}), that cause the resulting distribution to also land in $R$. 

\textbf{Case 1. $C$ is even}. Set $w = 1$ and leave $z$ as a variable in $\pe$. Applying the resulting transition matrix to
$(p_1, p_2, p_3, p_4)$ yields the distribution
$$\left(\frac{p_1}{2} + zp_2,\,
\frac{p_3}{2} + p_4,\,
\frac{p_1}{2} + (1-z)p_2,\,
\frac{p_3}{2}\right).$$
Let $z = \frac{1}{p_2}(\frac23 - p_4 - \frac{p_1+p_3}{2})$ (this is the value of $z$ that makes the first two components sum to $\frac23$). To show that it is valid to set $z$ to this value, we have to show that $z \in [0,1]$ (since this is the feasible range for $z$ in the $\pe$ chain). First, $z \geq 0$ because $\frac{p_1 + p_3}{2} + p_4 \leq \frac23$.\footnote{Since $p_1 + p_2 = \frac23$, $\frac{p_1 + p_3}{2} + p_4$ is maximized when $p_1 = \frac23$, $p_2 = p_3 = 0$, and $p_4 = \frac13$.} On the other hand, $z \leq 1$ is equivalent to $$\frac{p_1}{2} + p_2 + \frac{p_3}{2} + p_4 \geq \frac23.$$ Plugging in $p_1 = \frac23 - p_2$ and $p_3 = \frac13 - p_4$, this becomes equivalent to $p_2 + p_4 \geq \frac13$, which is a constraint defining $R$. Thus it is valid to set $z$ to this value. Plugging in this value for $z$ and using $p_3=\frac13 - p_4$, the resulting distribution becomes
$$(q_1, q_2, q_3, q_4) := \left(\frac12 - \frac{p_4}{2}, \,
\frac{1}{6} + \frac{p_4}{2}, \, 
\frac{1}{6} + \frac{p_4}{2}, \, 
\frac{1}{6} - \frac{p_4}{2} \right)$$
It is easy to check that both $(q_1, q_2, q_3, q_4)$ and $(q_2, q_1, q_3, q_4)$ lie in $R$.


\textbf{Case 2. $C$ is odd}. 
Set $x = y = z = w = \frac23$ in $\po$.
Applying the resulting transition matrix to $(p_1, p_2, p_3, p_4)$ yields the distribution
$$
\left( 
\frac23 p_1 + \frac23 p_2, \,
\frac23 p_3 + \frac23 p_4, \,
\frac13 p_1 + \frac13 p_2, \, 
\frac13 p_3 + \frac13 p_4
\right)
= 
\left(
\frac49,\, \frac29,\, \frac29,\, \frac19
\right).
$$
It is easy to check that both $(\frac49, \frac29, \frac29, \frac19)$ and $(\frac29, \frac49, \frac29, \frac19)$ lie in $R$.
\end{proof}
}

\red{
\begin{proof}[Proof of \Cref{thm:region_necessary}]
The result follows from the following two statements.
\begin{enumerate}
    \item Any feasible distribution must have its first two coordinates summing to exactly $\frac23$.
    \item Given a feasible distribution whose first two coordinates sum to $\frac23$, it must in fact be in $R$. 
\end{enumerate} 
We will now prove these two statements.

\emph{Proof of Statement 1.} To show that any feasible distribution must have its first two coordinates summing to exactly $\frac23$, consider a general distribution $(p_1, p_2, p_3, p_4)$. Clearly in order to obtain a $\frac43$-approximation, we must have $p_1 + p_2 \geq \frac23$.\footnote{Since states 1 and 2 use each edge $\frac12$ of the time and states 3 and 4 use each edge once in expectation, $p_1 + p_2 < \frac23$ would imply each edge is used strictly more than $1 \cdot \frac13 + \frac12 \cdot \frac23 = \frac23$ times in expectation.} Thus we just need to show $p_1 + p_2$ cannot be strictly larger than $\frac23$. To prove this, suppose  $p_1 + p_2 > \frac23$. We will obtain a contradiction by applying $\pe$ \emph{twice}. 
Applying $\pe$ once with $z = z_1$ and $w = w_1$ (for some $z_1 \in [0,1], w_1 \in [0,1]$) to $(p_1, p_2, p_3, p_4)$, we get the distribution 
$$(q_1, q_2, q_3, q_4) := \left(\frac{p_1}{2} + p_2 z_1, \frac{p_3}{2} + p_4w_1, \frac{p_1}{2} + p_2(1 - z_1), \frac{p_3}{2} + p_4(1-w_1)\right).$$ In particular, note that $q_1 + q_3 = p_1 + p_2 > \frac23$. Applying $\pe$ a second time to $(q_1, q_2, q_3, q_4)$, with $z = z_2$ and $w = w_2$, we get a distribution whose first two coordinates sum to
$$\left(\frac{q_1}{2} + q_2z_2\right)+ \left(\frac{q_3}{2} + q_4w_2\right) \leq \frac{q_1+q_3}{2} + q_2 + q_4 = 1 - \frac{q_1+q_3}{2} < \frac23.$$
Since the first two coordinates of this distribution sum to strictly less than $\frac23$, it cannot give a $\frac43$-approximation.

\emph{Proof of Statement 2.} Having just shown that any feasible distribution must have its first two coordinates summing to exactly $\frac23$, we now show that any such distribution must in fact lie in $R$. Consider a general distribution whose first two coordinates sum to $\frac23$; we can write it as $(p_1, p_2, p_3, p_4)$ where $p_1 + p_2 = \frac23$. To show this point lies in $R$, we just need to show that $p_2 + p_4 \geq \frac13$. 
Applying $\pe$ to the input $(p_1, p_2, p_3, p_4)$, we obtain
$$\left( 
\frac12 p_1 + zp_2,\,
\frac12 p_3 + w p_4, \,
\frac12 p_1 + (1-z)p_2, \,
\frac12 p_3 + (1-w)p_4
\right).$$
We need the first two components to sum to $\frac23$, which means $\frac12 p_1 + zp_2 + \frac12 p_3 + w p_4 = \frac23$. Plugging in $x,w \leq 1$, we get $\frac12 p_1 + p_2 + \frac12 p_3 + p_4 \geq \frac23$.
Finally, using $p_1 + p_2 + p_3 + p_4 = 1$, we obtain $p_2 + p_4 \geq \frac13$.

\end{proof}
}

\begin{remark}
    Algorithm \ref{alg:main} can now be modified to use any initial distribution $\pv \in R$, not just $(\frac49, \frac29, \frac29, \frac19)$. To do this, simply begin by sampling edges entering $V \setminus \{r\}$ according to $\pv$. Then, given the edges entering a parent cycle cut, connect up its children using the rules given in the proof of \Cref{prop:mcs}, according to the transition probabilities given in the proof of \Cref{thm:region_sufficient}. 
\end{remark}

\begin{remark}
The "tightness" of the feasible region is respect to our Markov chains in \Cref{fig:general_mcs}. It is possible that there are other patterns / Markov chains that would give rise a larger feasible region. 
\end{remark}



\section{Conclusion and Open Questions}

As discussed at the end of the introduction, our result leads to several interesting open questions.  One such open question is whether our result extends to the case of cycle cuts for non-half-integral solutions. We believe this to be possible through a more refined understanding of the patterns that result from considering non-half-integral solutions.

Clearly a better understanding of what happens in the case of degree cuts is needed to make substantial progress on the overall half-integral case.  As discussed in the introduction, we think it is possible to improve incrementally on the 1.4983-approximation of Gupta et al. \cite{GuptaLLMNS22} by using a combination of ideas from this paper with a few other small improvements. Recall that in a degree cut, each vertex has degree four, there are no parallel edges, and every non-trivial cut has at least six edges in it.  Ideally one would be able to show that any distribution on a parent cut lying in the feasible region of \cref{fig:region} could be used to induce a distribution on patterns of the children of the degree cut in a subregion of the feasible region with each edge used at most 2/3 of the time; such a result would lead immediately to a 4/3 integrality gap for half-integral instances.  

\iftoggle{notblinded}{
\subsection*{Acknowledgments}

The first and third authors would like to thank Anke van Zuylen for early discussions on this problem.  The first and third authors were supported in part by NSF grant CCF-2007009.  The first author was also supported by NSERC fellowship PGSD3-532673-2019. The second author was supported in part by NSF grants DGE-1762114, CCF-1813135, and CCF-1552097. \red{We would like to thank Martin Drees for his helpful suggestions that allowed us to simplify the proof of the main result.}

}{}

\printbibliography

@string{stoc20 = "Proceedings of the 52nd Annual {ACM} Symposium on the
	Theory of Computing"}

@string{stoc21 = "Proceedings of the 53rd Annual {ACM} Symposium on the
	Theory of Computing"}

@string{a = "Algorithmica"}

@string{c = "Combinatorica"}

@string{comb = "Combinatorica"}

@string{ipl = "Information Processing Letters"}

@string{jalg = "Journal of Algorithms"}

@string{jacm = "Journal of the {ACM}"}

@string{mor = "Mathematics of Operations Research"}

@string{mp = "Mathematical Programming"}

@string{or = "Operations Research"}

@article{BenoitB08,
author  = "Genevi\`eve Benoit and Sylvia Boyd",
title   = "Finding the exact integrality gap for small traveling salesman problems",
journal = mor,
volume  = 33,
pages   = "921--931",
year    = 2008}

@article{BoydC11,
author  = "Sylvia Boyd and Robert Carr",
title   = "Finding Low Cost {TSP} and 2-matching Solutions Using Certain Half-Integer Subtour Vertices",
journal = "Discrete Optimization",
volume = 8,
pages   = "525--539",
year    = 2011}

@article{BoydS21,
author = "Sylvia Boyd and Andr\'as Seb\H{o}",
title   = "The salesman's improved tours for fundamental classes",
journal  = mp,
volume = 186,
pages = "289--307",
year = 2021}

@techreport{Christofides76,
author  = "Nicos Christofides",
title   = "Worst Case Analysis of a New Heuristic for the Traveling Salesman
       Problem",
institution = "Graduate School of Industrial Administration,
           Carnegie-Mellon University",
type    = "Report",
number  = 388,
address = "Pittsburgh, PA",
year    = 1976}

@article{DantzigFJ54,
author  = "G. Dantzig and R. Fulkerson and S. Johnson",
title   = "Solution of a Large-Scale Traveling-Salesman Problem",
journal = or,
volume  = 2,
pages   = "393--410",
year    = 1954}

@article{Fleischer99,
author  = "Lisa Fleischer",
title   = "Building Chain and Cactus Representations of All Minimum Cuts from {H}ao-{O}rlin in the Same Asymptotic Run Time",
journal = jalg,
volume = 33,
pages   = "51--72",
year    = 1999}

@article{Goemans95,
author  = "Michel X. Goemans",
title   = "Worst-case comparison of valid inequalities for the {TSP}",
journal = mp,
volume  = 69,
pages   = "335--349",
year    = 1995}

@inproceedings{GuptaLLMNS22,
author  = "Anupam Gupta and Euiwoong Lee and Jason Li and Marcin Mucha and Heather Newman and Sherry Sarkar",
title   = "Matroid-Based {TSP} Rounding for Half-Integral Solutions",
booktitle = "Integer Programming and Combinatorial Optimization",
editor="Aardal, Karen and Sanit{\`a}, Laura",
series = "Lecture Notes in Computer Science",
number  = 13265,
pages   = "305--318",
year    = 2022,
note    = "See also \url{https://arxiv.org/pdf/2111.09290.pdf}"}

@InProceedings{HaddadanN19,
  author =	{Arash Haddadan and Alantha Newman},
  title =	{{Towards Improving Christofides Algorithm for Half-Integer TSP}},
  booktitle =	{27th Annual European Symposium on Algorithms (ESA 2019)},
  pages =	{56:1--56:12},
  series =	{Leibniz International Proceedings in Informatics (LIPIcs)},
  ISBN =	{978-3-95977-124-5},
  ISSN =	{1868-8969},
  year =	{2019},
  volume =	{144},
  editor =	{Michael A. Bender and Ola Svensson and Grzegorz Herman},
  publisher =	{Schloss Dagstuhl--Leibniz-Zentrum fuer Informatik},
  address =	{Dagstuhl, Germany},
  note="See also \url{https://arxiv.org/pdf/1907.02120v3.pdf}"}

@article{HeldK71,
author  = "Michael Held and Richard M. Karp",
title   = "The Traveling-Salesman Problem and Minimum Spanning Trees",
journal = or,
volume  = 18,
pages   = "1138--1162",
year    = 1971}

@incollection{JohnsonM07,
author  = "David S. Johnson and Lyle A. McGeoch",
title   = "Experimental Analysis of Heuristics for the {STSP}",
booktitle = "The Traveling Salesman Problem and its Variations",
editor  = "Gregory Gutin and Abraham P. Punnen",
pages   = "369--444",
publisher = "Springer",
year    = 2007}

@article{MomkeS16,
author  = {Tobias M\"omke and Ola Svensson},
title   = "Removing and Adding Edges for the Traveling Salesman Problem",
journal = jacm,
volume = 63,
year = 2016,
note = "Article 2"}

@article{SchalekampWvZ14,
  author    = {Frans Schalekamp and
               David P. Williamson and
               Anke {van Zuylen}},
  title     = {2-Matchings, the Traveling Salesman Problem, and the Subtour {LP:}
               {A} Proof of the {B}oyd-{C}arr Conjecture},
  journal   = mor,
  volume    = {39},
  number    = {2},
  pages     = {403--417},
  year      = {2014}}

@article{SeboV14,
author= {Andr\'as Seb\H{o} and Jens Vygen},
title={Shorter tours by nicer ears: 7/5-approximation for the graph-{TSP}, 3/2 for the path version, and 4/3 for two-edge-connected subgraphs},
journal = comb,
pages = "1--34",
year={2014}
}

@article{Serdyukov78,
author	= "A. Serdyukov",
title	= "On some extremal walks in graphs",
journal	= "Upravlyaemye Sistemy",
volume	= 17,
pages	= "76--79",
year	= 1978}

@article{ShmoysW90,
author  = "David B. Shmoys and David P. Williamson",
title   = "Analyzing the {H}eld-{K}arp {TSP} bound: A monotonicity property
       with application",
journal = ipl,
volume  = 35,
pages   = "281--285",
year    = 1990}

@article{Wolsey80,
author  = "L. A. Wolsey",
title   = "Heuristic analysis, linear programming and branch and bound",
journal = "Mathematical Programming Study",
volume  = 13,
pages   = "121--134",
year    = 1980}

@inproceedings{KKO20,
  author    = {Anna R. Karlin and
               Nathan Klein and
               Shayan {Oveis Gharan}},
  editor    = {Konstantin Makarychev and
               Yury Makarychev and
               Madhur Tulsiani and
               Gautam Kamath and
               Julia Chuzhoy},
  title     = {An improved approximation algorithm for {TSP} in the half integral
               case},
  booktitle = stoc20,
  pages     = {28--39},
  publisher = {{ACM}},
  year      = {2020},
}

@INPROCEEDINGS {KKO21b,
author = {A. Karlin and N. Klein and S. {Oveis Gharan}},
booktitle = {2022 IEEE 63rd Annual Symposium on Foundations of Computer Science (FOCS)},
title = {A (Slightly) Improved Bound on the Integrality Gap of the Subtour LP for TSP},
year = {2022},
volume = {},
issn = {},
pages = {832-843},
doi = {10.1109/FOCS54457.2022.00084},
url = {https://doi.ieeecomputersociety.org/10.1109/FOCS54457.2022.00084},
publisher = {IEEE Computer Society},
address = {Los Alamitos, CA, USA},
month = {nov}
}

@inproceedings{KKO21,
      title={A (Slightly) Improved Approximation Algorithm for Metric TSP}, 
      author={Anna R. Karlin and Nathan Klein and Shayan {Oveis Gharan}},
      year={2021},
      booktitle=stoc21,
    pages="32--45",
      publisher={ACM},
}

@techreport{FF09,
	Author = {Fleiner, Tam{\'a}s and Frank, Andr{\'a}s},
	Institution = {Egerv{\'a}ry Research Group, Budapest},
	Note = {{\tt www.cs.elte.hu/egres}},
	Number = {QP-2009-03},
	Title = {A quick proof for the cactus representation of mincuts},
	Year = {2009}}

@book{Williamson19,
author  = "David P. Williamson",
title   = "Network Flow Algorithms",
publisher = "Cambridge University Press",
address = "Cambridge, United Kingdom",
year    = 2019}

@book{grotschel2012geometric,
  title={Geometric algorithms and combinatorial optimization},
  author={Gr{\"o}tschel, Martin and Lov{\'a}sz, L{\'a}szl{\'o} and Schrijver, Alexander},
  volume={2},
  year={2012},
  publisher={Springer Science \& Business Media}
}
\end{document}